\documentclass[10pt,journal]{IEEEtran}

\usepackage{amssymb}
\usepackage{amsmath}
\usepackage{cite}
\usepackage{url}
\usepackage{xcolor}
\usepackage{cite,graphicx,amsmath,amssymb}
\usepackage{fancyhdr}
\usepackage{mdwmath}
\usepackage{mdwtab}
\usepackage{caption}
\usepackage{amsthm}
\usepackage{setspace}
\usepackage{hyperref}
\usepackage{algorithm}
\usepackage{algorithmic}
\usepackage{multirow}
\usepackage{makecell}
\usepackage{mathtools}
\usepackage{subcaption}
\usepackage{bm}
\usepackage{tikz}
\usepackage{xurl}
\usepackage{stfloats}
\usepackage{mathrsfs}

\usepackage{booktabs}
\usepackage{multirow}
\usepackage{siunitx}
\usepackage{balance}

\hypersetup{colorlinks=true,
linkcolor=blue,
citecolor=blue,      
urlcolor=black,
}

\graphicspath{{./src/img/}}

\newtheorem{remark}{Remark}
\newtheorem{theorem}{Theorem}

\newtheorem{lemma}{Lemma}

\newtheorem{corollary}{Corollary}

\newtheorem{proposition}{Proposition}

\allowdisplaybreaks
\NewDocumentCommand{\multiubrace}{mmm}
 {
  \egreg_multiubrace:nnn {#1} {#2} {#3}
 }
 \def\mathbi#1{\textbf{\em #1}} 

\title{Joint DOA and Attitude Sensing Based on Tri-Polarized Continuous Aperture Array}

\author{Haonan Si,~\IEEEmembership{Graduate Student Member, IEEE}, Zhaolin Wang, ~\IEEEmembership{Member, IEEE},\\	Xiansheng Guo, ~\IEEEmembership{Senior Member, IEEE}, Jin Zhang, ~\IEEEmembership{Member, IEEE}, Yuanwei Liu, ~\IEEEmembership{Fellow, IEEE}
	\thanks{This work was supported by the National Natural Science Foundation of China under Grant 62171086. \emph{(Corresponding author: Xiansheng Guo.)}}
	\thanks{Haonan Si and Xiansheng Guo are with the Department of Electronic Engineering, University of Electronic Science and Technology of China, Chengdu, 611731, China (e-mail: sihaonan@std.uestc.edu.cn,  xsguo@uestc.edu.cn).}
	\thanks{Zhaolin Wang and Yuanwei Liu are with the Department of Electrical and Electronic Engineering, The University of Hong Kong, Hong Kong  (e-mail: zhaolin.wang@hku.hk, yuanwei@hku.hk).}
	\thanks{Jin Zhang is with the School of Electronic Engineering and Computer Science, Queen Mary University of London, London E1 4NS, U.K. (e-mail: jin.zhang@qmul.ac.uk).}
}

\begin{document}

\maketitle
\begin{abstract}
	This paper investigates joint direction-of-arrival (DOA) and attitude sensing using tri-polarized continuous aperture arrays (CAPAs). By employing electromagnetic (EM) information theory, the spatially continuous received signals in tri-polarized CAPA are modeled, thereby enabling accurate DOA and attitude estimation. To facilitate subspace decomposition for continuous operators, an equivalent continuous–discrete transformation technique is developed. Moreover, both self- and cross-covariances of tri-polarized signals are exploited to construct a tri-polarized spectrum, significantly enhancing DOA estimation performance. Theoretical analyses reveal that the identifiability of attitude information fundamentally depends on the availability of prior target snapshots. Accordingly, two attitude estimation algorithms are proposed: one capable of estimating partial attitude information without prior knowledge, and the other achieving full attitude estimation when such knowledge is available. Numerical results demonstrate the feasibility and superiority of the proposed framework.
\end{abstract}


\section{Introduction} \label{sec:intro}

\IEEEPARstart{T}he evolution towards 6G networks imposes higher requirements on sensing systems, particularly in terms of resolution and accuracy. Such capabilities are essential for a wide range of applications, including position-based services, multi-agent networks, integrated sensing and communication, and digital twin systems \cite{10816694,10918620,10255711,10906057}. The sensing performance, however, is fundamentally constrained by the underlying antenna array paradigm. The introduction of multi-input multi-output (MIMO) technology has significantly expanded the spatial degrees of freedom (DoF), thereby enabling substantial improvements in sensing resolution and estimation accuracy \cite{7533517,6922542}. 

Most existing sensing algorithms, however, are developed under the far-field uniform plane-wave assumption, which stems from a physically inconsistent channel modeling approach \cite{10678869}. Under this assumption, the received signals are simplified as transmitting from point targets in free space, such that only directional information (i.e., direction-of-arrival, DOA) can be extracted. In contrast, practical electromagnetic (EM) propagation inherently depends not only on the target’s location but also on its orientation (or attitude) \cite{yang20153}, which modulates the propagation channel and renders attitude estimation feasible. This observation highlights the limitations of conventional DOA-only frameworks and motivates the development of joint DOA and attitude sensing approaches for more general and fine-grained scenarios.

\subsection{Prior Works}

EM information theory \cite{zhu2024electromagnetic} has emerged as an effective framework for modeling radiating signals from the fundamental Maxwell’s equations. Building on this theory, the authors of  \cite{miller2000communicating} employed eigenfunction analysis to demonstrate that the spatial DoF of a holograph MIMO is proportional to its physical volume. Subsequent studies \cite{poon2005degrees,dardari2020communicating} extended this result to more general scenarios and derived more comprehensive analytical expressions. These modeling insights have shifted research attention toward EM information theory-based sensing system analyses. For example, the authors of  \cite{zhang2022holographic} proposed a beamforming scheme that maximizes beampattern gain in the target direction, thereby improving sensing accuracy. Similarly, the authors of  \cite{d2022cramer} analyzed the performance limits of holograph-based sensing systems under both known and unknown target orientations. However, these works primarily focus on system-level optimization and performance analysis, while offering limited contributions to detailed algorithmic design, particularly in scenarios involving multiple targets.

To further unleash the potential of sensing systems, integrating an ever-larger number of antenna elements within a constrained physical aperture has become a natural evolution. Following this trend, wireless systems are transitioning from massive MIMO \cite{9757375} to gigantic MIMO \cite{11026007}, paving the way for the next generation of networks. More recently, the concept of continuous aperture antenna (CAPA) \cite{11095329} has emerged as a transformative paradigm, moving beyond discrete arrays toward realizing a physically continuous EM aperture. Owing to its extraordinary spatial DoF, CAPA has attracted significant research interest in both wireless communications and sensing \cite{11095329,11015930}. Nevertheless, the unique paradigm of continuous EM radiation introduces both opportunities and challenges, especially with respect to continuous EM modeling and sensing algorithm design.

To fill this gap, the authors of \cite{si2025doa} developed a CAPA-based MUSIC algorithm for direction-of-arrival (DOA) estimation, achieving significant performance improvements in both theoretical derivations and simulation studies. Moreover, the authors of \cite{jiang2024cram} proposed a maximum likelihood estimation (MLE) method for near-field target localization, where the continuous source current is optimized by minimizing the Cramér-Rao lower bound (CRLB). Notably, the sensing objectives in these studies are primarily limited to location (near-field) and DOA (far-field) estimation, where the target is simplified as a point in free space. In contrast, practical applications often demand the extraction of more fine-grained target attributes (such as attitude information) from the sensing system.

Currently, research on attitude sensing remains in its early stage. As analyzed in \cite{10678869}, tri-polarized antennas are capable of simultaneously receiving independent polarized signals along three orthogonal directions, as expected to substantially enhance sensing capability of CAPA systems. The authors of \cite{10556596} proposed a near-field positioning and attitude sensing algorithm with closed-form, low-complexity solutions. However, this method has two main limitations: 1) it assumes a single target with prior knowledge of the target’s current and length; 2) it refines the tri-polarized signal into a scalar one via the Poynting vector, which inevitably discards part of the target information and constrains the sensing performance.  These limitations motivate us to develop a joint DOA and attitude sensing algorithm for tri-polarized CAPA systems under more general and practical scenarios.

\subsection{Motivation and Contributions}

As discussed above, despite the promising DoF and potential performance gains offered by CAPA, algorithmic design for CAPA-based sensing remains largely underexplored, and the spatially continuous aperture poses significant challenges for target sensing. Moreover, most existing CAPA algorithms focus exclusively on DOA estimation and typically assume the presence of a single target. In practice, however, attitude information is also of great importance. To address this challenge, several insightful studies have exploited EM information theory to model received signals, which provide richer sensing information and enable attitude estimation \cite{10556596, 8240645}. Nevertheless, these methods often rely on restrictive assumptions and neglect the rich information contained in tri-polarized received signals.

To overcome these challenges, we propose a joint DOA and attitude sensing algorithm for CAPA systems under more general scenarios. Specifically, we develop an equivalent eigenvalue decomposition method to address the difficulty of subspace decomposition with spatially continuous apertures. Based on this, both self- and cross-covariance matrices are leveraged to construct a tri-polarized spectrum, enabling DOA estimation with improved resolution and robustness. Subsequently, we design an attitude estimation algorithm for two cases: with and without a prior knowledge of target snapshots. The main contributions of this paper are summarized as follows:
\begin{itemize}
	\item We establish an EM information-theoretic model for the spatially continuous received signal in tri-polarized CAPA-based sensing. Building upon this model, we propose a unified framework for joint DOA and attitude estimation that accommodates more general and practical sensing scenarios.
	\item To tackle the challenge of subspace decomposition with spatially continuous signals, we develop an equivalent continuous–discrete transformation technique that enables subspace decomposition for continuous operators. Leveraging this formulation, we further exploit both self- and cross-covariances of tri-polarized signals to construct a novel tri-polarized spectrum, which significantly enhances the resolution and robustness of DOA estimation.
	\item Based on the DOA estimates, we conduct a theoretical identifiability analysis of target attitude information, demonstrating its fundamental dependence on the availability of prior target snapshots. Accordingly, we design tailored algorithms for both cases, with and without prior knowledge, to estimate partial or full attitude information, thereby extending CAPA sensing capabilities beyond conventional DOA estimation to more general and practical scenarios.
\end{itemize}

\subsection{Organization and Notations}
 The remainder of this paper is organized as follows. Section \ref{sec:model} presents the system model and problem formulation for tri-polarized CAPA systems. Section \ref{sec:algorithm} elaborates on the proposed DOA and attitude estimation algorithms. Section \ref{sec:results} presents simulation results of the proposed method as well as benchmark methods under various system settings. Section \ref{sec:conclusion} finally concludes this paper.

\textit{Notations:} Scalars, column vectors, matrices, and sets are denoted by $a$, $\mathbf{a}$, $\mathbf{A}$, and $\mathcal{A}$, respectively. The transpose, conjugate, Hermitian, and inverse of a matrix are denoted by $(\cdot)^\mathsf{T}$, $(\cdot)^\ast$, $(\cdot)^\mathsf{H}$, and $(\cdot)^{-1}$, respectively. The real and imaginary parts of a complex variable are written as $\Re{(\cdot)}$ and $\Im{(\cdot)}$. The spaces of real and complex $M\times N$ matrices are denoted by $\mathbb{R}^{M\times N}$ and $\mathbb{C}^{M\times N}$, respectively. The $(m,n)$-th entry of a matrix and the $m$-th entry of a vector are denoted by $[(\cdot)]_{m,n}$ and $[(\star)]_{m}$, respectively, while $[(\cdot)]_{m,:}$ and $[(\cdot)]_{:,n}$ denote the $m$-th row and $n$-th column of a matrix. The Euclidean norm of a vector $(\star)$ is written as $\|(\star)\|_2$. The Lebesgue measure of a set $\mathcal{S}$ is denoted by $|\mathcal{S}|$. The imaginary unit is represented by $\mathsf{j}=\sqrt{-1}$.
\section{System Model and Problem Formulation} \label{sec:model}

\subsection{System Description}
The system of this work is illustrated in Fig. \ref{system_model}, where a CAPA is placed on the XOY plane to receive the signal. The center of CAPA is set at the origin of the coordinate system and its two sides are parallel to the X- and Y-axes, with the length being $L_x$ and $L_y$, respectively. Hence, the CAPA aperture size is $L_x\times L_y$. An arbitrary point on the CAPA is formulated as $\mathbf{r}\in \mathcal{S}$ and $\mathcal{S}$ denotes the coordinate region of CAPA, which is specified by: 
\begin{align}
		\mathcal{S}=\left\{\mathbf{r} \triangleq\left[r_x, r_y, r_z\right]^{\mathsf{T}}:-\frac{L_x}{2} \leq r_x \leq \frac{L_x}{2},\right. \nonumber\\
		\left.-\frac{L_y}{2} \leq r_y \leq \frac{L_y}{2}, r_z=0\right\}.
\end{align}

As shown in Fig. \ref{system_model},  we consider the target of interest to be a \textit{Hertzian dipole}. Suppose there are $M$ targets in total, with the $m$-th Hertzian dipole (target)  of length $l_m$ and located at position $\mathbf{p}_m = [p_{m,x}, p_{m,y}, p_{m,z}]^\mathsf{T}$, where $m \in \mathcal{M} = \{1, 2, \ldots, M\}$. Here, $p_{m,x}$, $p_{m,y}$, and $p_{m,z}$ denote the coordinates along the X-, Y-, and Z-axes, respectively. More specifically, the $m$-th target is pointed in the direction of $\bar{\mathbf{q}}_m=q_{m,x}\bar{\mathbf{x}}+q_{m,y}\bar{\mathbf{y}}+q_{m,z}\bar{\mathbf{z}}$, with $\Vert \bar{\mathbf{q}}_m\Vert =1$, and $\bar{\mathbf{x}}$, $\bar{\mathbf{y}}$, and $\bar{\mathbf{z}}$ being the unit vectors along the X-, Y-, and Z-axes, respectively. Notably, $\bar{\mathbf{q}}_m$ represents the rotation of the target in the free space, i.e., the attitude information of the $m$-th target. The targets are assumed to be located in the radiating far-field region, i.e., $z_0> \frac{2D^2}{\lambda}$, where $z_0$ is the minimum distance from the CAPA to the target, $D$ is the aperture size, $\lambda$ denotes the wavelength, respectively.

\begin{figure}[t!]
    \centering
    \includegraphics[width=0.48\textwidth]{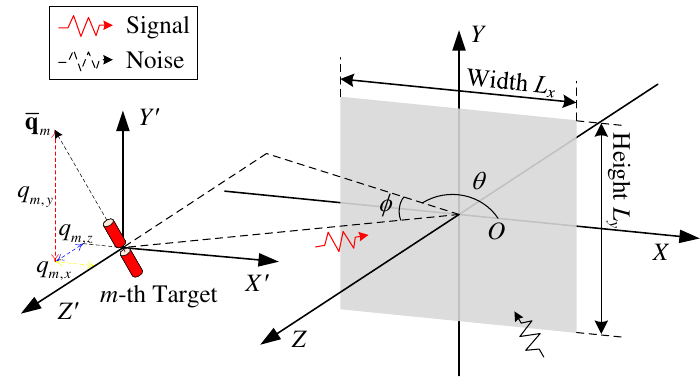}
    \caption{Illustration of the point-to-point CAPA-MIMO system.}
    \label{system_model}
\end{figure}

\subsection{Electromagnetic	Propagation Modeling}
 In this subsection, we derive the relationship between the target source currents and the received electric field on the CAPA using EM information theory.

 To model the relationship between the source current $\mathbf{J}(\mathbf{p}_v)$ and the electric field vector $\mathbf{e}(\mathbf{p}_v)$ at an arbitrary point $\mathbf{p}_v\in\mathbb{R}^3$, we involve the following inhomogeneous Helmholtz equation \cite{1386525}:
 
 \begin{align}\label{eq:1}
 	-\nabla_{\mathbf{p}_v} \times \nabla_{\mathbf{p}_v} \times \mathbf{e}\left(\mathbf{p}_v\right)+k^2 \mathbf{e}\left(\mathbf{p}_v\right)=\mathsf{j} k \eta \mathbf{J}\left(\mathbf{p}_v\right),
 \end{align}
 where $\nabla_{\mathbf{p}_v}$ denotes the curl operation with respect to $\mathbf{p}_v$, $k=\frac{2\pi}{\lambda}$ and $\eta$ represent the wavenumber and intrinsic impedance, respectively.
 
 Following the EM modeling methods \cite{10556596}, we utilize the Green's function to derive the inverse map of Eq. (\ref{eq:1}), with the electric field at the receive formulated as
 \begin{align}\label{Eq:3}
 	\mathbf{e}(\mathbf{r})=\int_{\mathbb{R}^2}\mathbf{G}(\mathbf{z})\mathbf{J}(\mathbf{p}_v)d\mathbf{p}_v,
 \end{align}
 where $\mathbf{z}=\mathbf{r}-\mathbf{p}_v$ denotes the radiation vector and $\mathbf{G}(\mathbf{z})$ is the tensor Green's function, which is widely used in electromagnetic analyses and can be formulated as
 \begin{align}
 	\mathbf{G}(\mathbf{z})=G_s(z)\left[\left(1+\frac{\mathsf{j}}{k z}-\frac{1}{k^2 z^2}\right) \mathbf{I}\right.\nonumber\\
 	\left.-\left(1+\frac{\mathsf{j} 3}{k z}-\frac{3}{k^2 z^2}\right) \bar{\mathbf{z}} \bar{\mathbf{z}}^{\mathsf{H}}\right],
 \end{align}
 where $z=||\mathbf{z}||_{2}$, $\bar{\mathbf{z}}=\frac{\mathbf{z}}{z}$, and $G_s(z)=\mathsf{j}\frac{\eta}{2\lambda z}e^{\mathsf{j}kz}$ is the scalar Green's function, respectively. In the far-field scenarios, the terms $\frac{\mathsf{j}}{kz}$ and $\frac{1}{k^2z^2}$ can be neglected \cite{10556596}, which further generates:
 \begin{align}\label{eq:5}
 	\mathbf{G}(\mathbf{z})\approx G_s(z)(\mathbf{I}-\bar{\mathbf{z}} \bar{\mathbf{z}}^\mathsf{H}).
 \end{align}
 Substituting Eq. (\ref{eq:5}) into Eq. (\ref{Eq:3}), it can be derived that:
 \begin{align}\label{eq:6}
 	\mathbf{e}(\mathbf{r})&=\int_{\mathbb{R}^2}G_s(z)(\mathbf{I}-\bar{\mathbf{z}} \bar{\mathbf{z}}^\mathsf{H})\mathbf{J}(\mathbf{p}_v)d\mathbf{p}_v\nonumber\\
 	&=\sum_{m=1}^{M}G_s(z_m)\left(\mathbf{I}-\bar{\mathbf{z}}_m\bar{\mathbf{z}}_m^\mathsf{T}\right) \int_{\mathbb{R}^2} \mathbf{J}(\mathbf{p}_v) d\mathbf{p}_v,
 \end{align}
 where $z_m=||\mathbf{z}_m||_2$, $\mathbf{z}_m=\mathbf{r}-\mathbf{p}_m$, and  $\bar{\mathbf{z}}_m=\frac{\mathbf{z}_m}{z_m}$. For the considered dipole target, the current $\mathbf{J}(\mathbf{p}_v)$ can be rewritten into $\mathbf{J}(\mathbf{p}_v)=I_m(t)l_m\delta(\mathbf{p}_v)\bar{\mathbf{q}}_m$, with $I_m(t)$ being the time-varying current level in the $m$-th dipole at time $t$. Subsequently, the received electric field at $\mathbf{r}$ and time $t$ can be rewritten as
 \begin{align}\label{eq:7}
 	\mathbf{e}(\mathbf{r},t)&=\sum_{m=1}^{M}G_s(z_m)I_m(t)l_m\left(\mathbf{I}-\bar{\mathbf{z}}_m \bar{\mathbf{z}}_m^\mathsf{T}\right)\int_{\mathbb{R}^2}  \delta(\mathbf{p}_v) d\mathbf{p}_v\nonumber\\
 	&=\sum_{m=1}^{M}G_s(z_m)I_m(t)l_m\underbrace{\left(\bar{\mathbf{q}}_m-\bar{\mathbf{z}}_m\bar{\mathbf{z}}_m^\mathsf{T}\bar{\mathbf{q}}_m \right)}_{\mathbf{v}_m},
 \end{align} 
 where $\mathbf{v}_m\triangleq\bar{\mathbf{q}}_m-\bar{\mathbf{z}}_m\bar{\mathbf{z}}_m^\mathsf{T}\bar{\mathbf{q}}_m\in\mathbb{R}^{3}$ is defined as the polarization vector.
 Notably, the orientation vector $\bar{\mathbf{z}}_m$ can be expressed as $\bar{\mathbf{z}}_m=\cos \theta_m\cos\phi_m \bar{\mathbf{x}}+\sin\theta_m\cos\phi_m\bar{\mathbf{y}}+\sin\phi_m\bar{\mathbf{z}}$, such that Eq. (\ref{eq:7}) can be rewritten as Eqs. (\ref{eq:8})-(\ref{eq:10}), where $e_x(\mathbf{r,t})\triangleq \mathbf{e}^\mathsf{T}(\mathbf{r},t)\bar{\mathbf{x}}$, $e_y(\mathbf{r,t})\triangleq \mathbf{e}^\mathsf{T}(\mathbf{r},t)\bar{\mathbf{y}}$, $e_z(\mathbf{r,t})\triangleq \mathbf{e}^\mathsf{T}(\mathbf{r},t)\bar{\mathbf{z}}$.
 
 \begin{figure*}[b]
 	\hrulefill
 	\begin{align}\label{eq:8}
 		&e_x(\mathbf{r},t)=\sum_{m=1}^{M}G_s(z_m)I_m(t)l_m\left[q_{m,x}-\left(\cos \theta_m\cos\phi_m q_{m,x}+\sin\theta_m\cos\phi_mq_{m,y}+\sin\phi_mq_{m,z}\right)\cos \theta_m\cos\phi_m\right],\\  \label{eq:9}
 		&e_y(\mathbf{r},t)=\sum_{m=1}^{M}G_s(z_m)I_m(t)l_m\left[q_{m,y}-\left(\cos \theta_m\cos\phi_m q_{m,x}+\sin\theta_m\cos\phi_mq_{m,y}+\sin\phi_mq_{m,z}\right)\sin\theta_m\cos\phi_m\right],\\ \label{eq:10}
 		&e_z(\mathbf{r},t)=\sum_{m=1}^{M}G_s(z_m)I_m(t)l_m\left[q_{m,z}-\left(\cos \theta_m\cos\phi_m q_{m,x}+\sin\theta_m\cos\phi_mq_{m,y}+\sin\phi_mq_{m,z}\right)\sin\phi_m\right].
 	\end{align}
 \end{figure*}

To facilitate the following signal modeling and algorithm design, we reformulate $G_s(z_m)$ by using the planar wave approximation. To begin with, $z_m$ can be reformulated as follows:
\begin{align}
	z_m&=\|\mathbf{r}-\mathbf{p}_m \|_2 \nonumber \\
	&=p_m\sqrt{1-2\frac{\mathbf{r}^\mathsf{T}\bar{\mathbf{p}}_m}{p_m}+\frac{r^2}{p_m^2}} \nonumber\\
	&\stackrel{(a)}{\approx} p_m-\mathbf{r}^\mathsf{T}\bar{\mathbf{p}}_m+\frac{\left[1+(\mathbf{r}^\mathsf{T}\bar{\mathbf{p}}_m)^2\right]r^2}{2p_m}, \label{eq:appr}
\end{align}
where $p_m=\|\mathbf{p}_m\|_2$, $r=\|\mathbf{r}\|_2$, $\bar{\mathbf{p}}_m=\mathbf{p}_m/p_m$, and approximation $(a)$ stems from the Taylor expansion the Taylor series expansion
$\sqrt{1+x} \approx 1 + \frac{x}{2} - \frac{x^2}{8}$,
which is valid for $|x| \ll 1$. By defining the wave vector $\mathbf{d}(\theta_m,\phi_m)=[\cos\theta_m\cos\phi_m,\sin\theta_m\cos\phi_m,\sin\phi_m]^\mathsf{T}$ and neglecting the 2-order terms, $G_s(z_m)$ can be further formulated as
\begin{align}\label{eq:12}
	G_s(z_m)\approx\mathsf{j}\frac{\eta}{2\lambda p_m}e^{\mathsf{j}kp_m}e^{-\mathsf{j}k\mathbf{r}^\mathsf{T}\mathbf{d}(\theta_m,\phi_m)}.
\end{align}
By substituting Eq. \eqref{eq:12} into \eqref{eq:7}, the electric field vector at time $t$ can be rewritten as
\begin{align}
	\mathbf{e}(\mathbf{r},t)=\sum_{m=1}^{M}s_m(t)a(\mathbf{r},\theta_m,\phi_m)\mathbf{v}_m,
\end{align}
where $s_m(t)=\mathsf{j}I_m(t)l_m\frac{\eta}{2\lambda p_m}e^{\mathsf{j}kp_m}$ denotes the snapshot signal of the $m$-th target and $a(\mathbf{r},\theta_m,\phi_m)=e^{-\mathsf{j}k\mathbf{r}^\mathsf{T}\mathbf{d}(\theta_m,\phi_m)}$.

\subsection{CAPA Receiver Modeling}
In this paper, we assume that the excited electric field is subject to a spatially uncorrelated zero-mean complex Gaussian process $n(\mathbf{r},t)$ \cite{10910020}, with the correlation function formulated as
\begin{align}
	&\mathbb{E}\left\{ \mathbf{n}(\mathbf{r}_1,t),\mathbf{n}^\mathsf{H}(\mathbf{r}_2,t)\right\}=\sigma^2\delta(\mathbf{r}_1-\mathbf{r}_2)\mathbf{I}_3,\nonumber\\
	&\forall \mathbf{r}_1,\mathbf{r}_2\in \mathcal{S},
\end{align}
where $\sigma^2$ denotes the noise covariance. To promote the following algorithm design, we equivalently transform the continuous aperture into an infinite number of infinitesimal non-overlapping receiving units, which is formulated as $\mathcal{S}=\bigcup_{n=1}^{N}\mathcal{S}_n$ and $N\rightarrow\infty$. Then, the voltage observed at the $n$-th unit can be obtained by integrating over the region as
\begin{align}
	\mathbf{x}_n(t)=&\int_{\mathcal{S}_n}\mathbf{e}(\mathbf{r},t)d\mathbf{r}\nonumber\\
	\approx&|\mathcal{S}_n|\left[\sum_{m=1}^{M}a(\mathbf{r}_n,\theta_m,\phi_m)s_m(t)\mathbf{v}_m+\mathbf{n}(\mathbf{r}_n,t)\right] \nonumber \\
	=&|\mathcal{S}_n|\left[\mathbf{V}^\mathsf{T}\mathbf{S}(t)\mathbf{a}(\mathbf{r}_n,\boldsymbol{\theta},\boldsymbol{\phi}) +\mathbf{n}(\mathbf{r}_n,t)\right],
\end{align}
where $n\in\mathcal{N}$, $\mathcal{N}=\left\{1,2,...,N\right\}$, $\mathbf{r}_n\in\mathcal{S}_n$ is an arbitrary point within the $n$-th region, 
\begin{align}
	&\mathbf{a}(\mathbf{r}_n,\boldsymbol{\theta},\boldsymbol{\phi})\nonumber\\
	=&\left[a(\mathbf{r}_n,{\theta}_1,{\phi}_1),a(\mathbf{r}_n,{\theta}_2,{\phi}_2),..., a(\mathbf{r}_n,{\theta}_M,{\phi}_M)\right]^\mathsf{T},
\end{align} $\mathbf{S}(t)=\text{diag}\left\{s_1(t),s_2(t),...,s_m(t)\right\}$, $\mathbf{V}=\left[\mathbf{v}_1,\mathbf{v}_2,...,\mathbf{v}_M\right]^\mathsf{T}\in\mathbb{R}^{M\times 3}$, $\boldsymbol{\theta}=\left[\theta_1,\theta_2,...,\theta_M\right]^\mathsf{T}$and $\boldsymbol{\phi}=\left[\phi_1,\phi_2,...,\phi_M\right]^\mathsf{T}$ denote the azimuth and elevation vectors, respectively.
Due to the fact that CAPA has theoretically infinite array units, we assume that the areas of these $N$ regions are the same, i.e., $\forall n\in\mathcal{N}$, $\lim_{N\rightarrow\infty}|\mathcal{S}_n|=A_0\rightarrow 0$. By concatenating  $\mathbf{x}_n(t)$, $\mathbf{a}(\mathbf{r}_n,\boldsymbol{\theta},\boldsymbol{\phi})$, and $\mathbf{n}(\mathbf{r}_n,t)$ across $n\in\mathcal{N}$, we derive the following matrices:
\begin{align}
	&\mathbf{X}(t)=\left[\mathbf{x}_1(t),\mathbf{x}_2(t),...,\mathbf{x}_N(t)\right]^\mathsf{T}\in\mathbb{C}^{N\times 3},\\
	&\mathbf{A}(\boldsymbol{\theta},\boldsymbol{\phi})=\left[\mathbf{a}(\mathbf{r}_1,\boldsymbol{\theta},\boldsymbol{\phi}),...,\mathbf{a}(\mathbf{r}_N,\boldsymbol{\theta},\boldsymbol{\phi})\right]^\mathsf{T}\in\mathbb{C}^{N\times M},\\
	&\mathbf{N}(t)=\left[\mathbf{n}(\mathbf{r}_1,t),\mathbf{n}(\mathbf{r}_2,t),...,\mathbf{n}(\mathbf{r}_N,t)\right]^\mathsf{T}\in\mathbb{C}^{N\times 3}.
\end{align}
Then, the received signal in CAPA can be formulated in a compact form:
\begin{align}\label{eq:20}
	\mathbf{X}(t)=\lim_{A_0\rightarrow 0}A_0 \left[\mathbf{A}(\boldsymbol{\theta},\boldsymbol{\phi})\mathbf{S}(t)\mathbf{V}+\mathbf{N}(t)\right].
\end{align}

Subsequently, the goal of this paper is to jointly estimate the orientation information $\left\{
\theta_m,\phi_m\right\}_{m=1}^{M}$ and attitude information $\left\{q_{m,x},q_{m,y},q_{m,z}\right\}_{m=1}^{M}$ from the CAPA receiver across $T$ snapshots $\left\{\mathbf{X}(t)\right\}_{t=1}^T$.
\begin{remark}
	\normalfont
	In this paper, the measurements of CAPA system are of infinite dimension ($N\rightarrow\infty$), rendering traditional DOA estimation methods such as MUSIC \cite{9784869} and semi-definite programming \cite{9695358} no longer applicable. Moreover, existing DOA estimation methods are based on single polarization models, limiting the potential estimation accuracy.
\end{remark}

\section{Joint DOA and Attitude Sensing Algorithm} \label{sec:algorithm}
In this section, we present the complete signal processing framework for tri-polarized CAPA. We begin with subspace decomposition using both self- and cross-covariance matrices to capture DOA information (\ref{III.A}), and then design an equivalent transformation–based algorithm to address the infinite-dimensional structure of CAPA (\ref{III.B}). We then provide theoretical analyses on the identifiability of the attitude vector (\ref{III.C}), followed by practical estimation methods for both general and special cases (\ref{III.D}). Finally, the computational complexity of the proposed scheme is discussed to assess its feasibility in real-time applications (\ref{III.E}).

\subsection{Subspace Decomposition for Tri-Polarized CAPA}\label{III.A}

Notably, the received signal $\mathbf{X}(t)$ consists of the signals along three polarization components: $\mathbf{X}(t)=[\mathbi{x}_x(t),\mathbi{x}_y(t),\mathbi{x}_z(t)]$,  with $\mathbi{x}_p(t)\in\mathbb{C}^{N}$ for $p\in\mathcal{P}$ and $\mathcal{P}=\{x,y,z\}$.  The statistical structure of $\mathbf{X}(t)$ can be fully characterized by its covariance matrices, where both the self-covariance $\mathbb{E}(\mathbi{x}_p\mathbi{x}_p^\mathsf{H})$ and cross-covariance $\mathbb{E}(\mathbi{x}_p\mathbi{x}_q^\mathsf{H})$ inherently encode the array manifold ($p,q\in\mathcal{P}, p\neq q$). Through subspace decomposition, these covariance matrices can be separated into signal and noise subspaces, from which the DOA information becomes identifiable. To facilitate the subsequent derivation, we analyze the problem in two cases: (i) self-covariance matrices and (ii) cross-covariance matrices.
\subsubsection{Self-covariance matrices} \label{subsubsec:1}

Firstly, we define the following covariance matrix of $\mathbi{x}_p(t)$:
\begin{align}
	\mathbf{R}^\mathbf{X}_{pp}&=\lim\limits_{A_0\rightarrow 0}\mathbb{E}(\mathbi{x}_p\mathbi{x}_p^\mathsf{H})\nonumber\\
	&=\lim\limits_{A_0\rightarrow 0}\mathbf{A}\mathbf{R}_{pp}^{\mathbf{S}}\mathbf{A}^\mathsf{H}+\sigma^2\mathbf{I}_N, p\in\mathcal{P},
\end{align}
where  $\mathbf{R}^{\mathbf{S}}_{pp}=\|\mathbi{v}_p\|_2^2\mathbb{E}(\mathbf{S}\mathbf{S}^\mathsf{H})$ denotes the covariance matrix of source signals modulated by  $\mathbi{v}_p$ and $\mathbi{v}_p=[\mathbf{V}]_{p,:}$ is the polarization vector. Notably, $\mathbf{R}^{\mathbf{S}}_{pp}$ is a positive definite matrix as long as $\|\mathbi{v}_p\|_2\neq 0$, 

On the other hand, $\mathbf{R}^{\mathbf{X}}_{pp}$ is a diagonal  matrix, which can be formulated by using eigenvalue decomposition:
\begin{align}
	\mathbf{R}^{\mathbf{X}}_{pp}=\lim\limits_{A_0\rightarrow 0}\mathbf{U}_{pp}\boldsymbol{\Lambda}_{pp}\mathbf{U}_{pp}^\mathsf{H},
\end{align}
where $\mathbf{U}_{pp}\in\mathbb{C}^{N\times N}$ is a unitary matrix consisting of eigenvectors of $\mathbf{R}^{\mathbf{X}}_{pp}$, which satisfies $\mathbf{U}_{pp}\mathbf{U}_{pp}^\mathsf{H}=\mathbf{I}_N$.  $\boldsymbol{\Lambda}_{pp}=\operatorname{diag}\{\lambda_1,\lambda_2,...,\lambda_N\}$ with $\lambda_n$ being the $n$-th eigenvalue and $n\in\mathcal{N}$. Hence, it can be derived that:
\begin{align}\label{eq:::23}
	\lim\limits_{A_0\rightarrow 0}\mathbf{U}_{pp}\boldsymbol{\Lambda}_{pp}\mathbf{U}_{pp}^\mathsf{H}=&\lim\limits_{A_0\rightarrow 0}\mathbf{A}\mathbf{R}^{\mathbf{S}}_{pp}\mathbf{A}^\mathsf{H}+\sigma^2\mathbf{I}_N\nonumber\\
	=&\lim\limits_{A_0\rightarrow 0}\mathbf{A}\mathbf{R}^{\mathbf{S}}_{pp}\mathbf{A}^\mathsf{H}+\sigma^2\mathbf{U}_{pp}\mathbf{U}_{pp}^\mathsf{H}.
\end{align}
Subsequently, one has $\mathbf{U}_{pp}(\boldsymbol{\Lambda}_{pp}-\sigma^2\mathbf{I}_N)\mathbf{U}_{pp}^\mathsf{H}=\mathbf{A}\mathbf{R}^{\mathbf{S}}_{pp}\mathbf{A}^\mathsf{H}$.  Apparently, the rank of $\mathbf{A}\mathbf{R}^{\mathbf{S}}_{pp}\mathbf{A}^\mathsf{H}$ is no bigger than $M$, rendering that $\boldsymbol{\Lambda}_{pp}-\sigma^2\mathbf{I}_N$ has a maximum amount of $M$ non-zero elements. To promote the following analyses, we define the sub-matrices as follows: $\mathbf{U}_{pp}=[\mathbf{U}_{pp,1},\mathbf{U}_{pp,2}]$ with $\mathbf{U}_{pp,1}\in \mathbb{C}^{N\times M}$ and $\mathbf{U}_{pp,2}\in \mathbb{C}^{N\times (N-M)}$ being the eigenvectors correponding to the signal and noise subspaces, respectively. Let $\boldsymbol{\Lambda}_{pp,1}=[\boldsymbol{\Lambda}_{pp}]_{1:M,1:M}$ and Eq. \eqref{eq:::23} can be rewritten into
\begin{align}
	\lim\limits_{A_0\rightarrow 0}\mathbf{U}_{pp,1}(\boldsymbol{\Lambda}_{pp,1}-\sigma^2\mathbf{I}_M)\mathbf{U}_{pp,1}^\mathsf{H}=\lim\limits_{A_0\rightarrow 0}\mathbf{A}\mathbf{R}^{\mathbf{S}}_{pp}\mathbf{A}^\mathsf{H}.
\end{align}
Based on the theoretical results in \cite{stoica2002music}, the following condition always holds given that $\mathbf{R}_{pp}^{\mathbf{S}}>\mathbf{0}$:
\begin{align}
	\operatorname{Span}(\mathbf{U}_{pp,1})=\operatorname{Span}(\mathbf{A}),
\end{align}
where $\operatorname{Span}(\star)$ denotes the space spanned by the linear vectors of matrix $\star$. Recalling that $\mathbf{U}_{pp}\mathbf{U}_{pp}^\mathsf{H}=\mathbf{I}_N$, one has $\mathbf{U}_{pp,1}\mathbf{U}_{pp,2}^\mathsf{H}=\mathbf{0}$ and the following equation holds:
\begin{align}\label{eq:r1}
	\lim\limits_{A_0\rightarrow 0}\mathbf{A}\mathbf{U}_{pp,2}^\mathsf{H}=\mathbf{0}.
\end{align}
Based on this result, the DOA parameters in $\mathbf{A}$ can be calculated. 

\subsubsection{Cross-covariance matrices} 
Firstly, we define the cross-covariance matrice between the $p$ and $q$ polarization directions:
\begin{align}
	\mathbf{R}^{\mathbf{X}}_{pq}&=\lim\limits_{A_0\rightarrow 0}\mathbb{E}(\mathbi{x}_p\mathbi{x}_q^\mathsf{H})\nonumber\\
	&=\lim\limits_{A_0\rightarrow 0}\mathbf{A}\mathbf{R}_{pq}^{\mathbf{S}}\mathbf{A}^\mathsf{H}, p\neq q, p,q\in\mathcal{P},
\end{align}
where  $\mathbf{R}^{\mathbf{S}}_{pq}=\mathbb{E}(\mathbf{S}\mathbi{v}_p\mathbi{v}_q^\mathsf{H}\mathbf{S}^\mathsf{H})$ denotes the variance matrix of $\mathbf{S}$ modulated by $\mathbi{v}_p\mathbi{v}_q^\mathsf{H}$, which is a positive definite matrix as long as $\|\mathbi{v}_p\|_2\neq 0$ and $\|\mathbi{v}_q\|_2\neq 0$. On the other hand, $\mathbf{R}^{\mathbf{X}}_{pp}$ can be formulated by using eigenvalue decomposition:
\begin{align}
	\mathbf{R}^{\mathbf{X}}_{pq}=\lim\limits_{A_0\rightarrow 0}\mathbf{U}_{pq}\boldsymbol{\Lambda}_{pq}\mathbf{V}_{pq}^\mathsf{H}=\lim\limits_{A_0\rightarrow 0}\mathbf{A}\mathbf{R}_{pq}^{\mathbf{S}}\mathbf{A}^\mathsf{H},
\end{align}
Similarly, it can be observed that $\boldsymbol{\Lambda}_{pq}$ has a maximum amount of $M$ non-zero elements, such that we define the sub-matrices: $\mathbf{U}_{pq}=[\mathbf{U}_{pq,1},\mathbf{U}_{pq,2}]$, $\mathbf{V}_{pq}=[\mathbf{V}_{pq,1},\mathbf{V}_{pq,2}]$ with $\mathbf{U}_{pq,1},\mathbf{V}_{pq,1}\in \mathbb{C}^{N\times M}$ and  $\boldsymbol{\Lambda}_{pq,1}=[\boldsymbol{\Lambda}_{pq}]_{1:M,1:M}$.  Then, it can be derived that:
\begin{align}
	\lim\limits_{A_0\rightarrow 0}\mathbf{U}_{pq,1}\boldsymbol{\Lambda}_{pq,1}\mathbf{V}_{pq,1}^\mathsf{H}=\lim\limits_{A_0\rightarrow 0}\mathbf{A}\mathbf{R}_{pq}^{\mathbf{S}}\mathbf{A}^\mathsf{H}.
\end{align}

Then, the following lemma is involved to promote the analyses.

\begin{lemma}\label{lemma:1}
	\normalfont Given that $\mathbf{R}^{\mathbf{S}}_{pq}$ is a positive definite matrix, the null space of $\mathbf{U}_{pq,1}$ belongs to that of $\mathbf{A}$, i.e., $ \operatorname{Null}(\mathbf{U}_{pq,1})\in\operatorname{Null}(\mathbf{A})$.
\end{lemma}

\begin{proof}
	Take any $\mathbf{y}\in \operatorname{Null}(\mathbf{U}_{pq,1})$, left-multiplying the identity
	\(
	\mathbf{U}_{pq,1}\boldsymbol{\Lambda}_{pq,1}\mathbf{V}_{pq,1}^{\mathrm H}
	=
	\mathbf{R}^{\mathbf{S}}_{pq}\mathbf{A}^{\mathsf H}
	\)
	by \(\mathbf{y}^{\mathsf H}\) yields
	\begin{align}
		\mathbf{y}^\mathsf{H}\mathbf{U}_{pq,1}\boldsymbol{\Lambda}_{pq,1}\mathbf{U}_{pq,1}^\mathsf{H}=\mathbf{y}^\mathsf{H}\mathbf{A}\mathbf{R}^{\mathbf{S}}_{pq}\mathbf{A}^\mathsf{H}=\mathbf{0}.
	\end{align}
	Hence, it can be derived that:
	\begin{align}\label{eq:lema1}
		\left(\mathbf{y}^\mathsf{T}\mathbf{A}\right)\mathbf{R}^{\mathbf{S}}_{pq}\left(\mathbf{y}^\mathsf{T}\mathbf{A}\right)^\mathsf{H}=\mathbf{0}.
	\end{align}
	Since $\mathbf{R}^{\mathbf{S}}_{pq}> \mathbf{0}$ is a positive definite matrix, Eq. \eqref{eq:lema1} holds if and only if $\mathbf{y}^\mathsf{T}\mathbf{A}=\mathbf{0}$, i.e., $\mathbf{y}\in\operatorname{Null}(\mathbf{A})$. This completes the proof of \textbf{Lemma} \ref{lemma:1}.
\end{proof}

Then, similar to Eq. 
\eqref{eq:r1}, the following equation holds:
\begin{align}\label{eq:r2}
	\lim\limits_{A_0\rightarrow 0}\mathbf{A}\mathbf{U}_{pq,2}^\mathsf{H}=\mathbf{0}.
\end{align}

\begin{remark}
	\normalfont Notably, $\mathbf{A}$ can be recovered from each covariance matrix $\mathbf{R}_{pq}^\mathbf{X}$ for $p,q\in\mathcal{P}$. Consequently, up to nine independent equations are available for DOA estimation, which distinguishes the studied tri-polarized CAPA from conventional arrays that capture only a single polarization. 
\end{remark}

\begin{remark}
	\normalfont
	The derivation of Eqs. \eqref{eq:r1} and \eqref{eq:r2} relies on the assumption that $|\mathbi{v}_p| \neq 0$ for some $p \in \mathcal{P}$; equivalently, the projection of $\mathbf{v}_m$ onto the plane $\mathbf{z}_m^\perp$ must have a nonzero component along at least one axis $p$. Therefore, as long as $\mathbf{v}_m \nparallel \mathbf{z}_m$, there exists at least one $p \in \mathcal{P}$ such that the DOA remains identifiable.
\end{remark}

In the next, we will delve into the derivation of $\mathbf{U}_{pq,2}$ in practical applications.

\subsection{DOA Estimation Algorithm Design}\label{III.B}
Typically, the covariance matrix between received signals along axes $p$ and $q$ is estimated by averging $\mathbi{x}_p(t)\mathbi{x}_q^\mathsf{H}(t)$ over $T$ snapshots, which is formulated as
\begin{align}
	\hat{\mathbf{R}}_{pq}^{\mathbf{X}}&=\lim\limits_{A_0\rightarrow 0}\frac{1}{T}\sum_{t=1}^{T}\mathbi{x}_p(t)\mathbi{x}_q^\mathsf{H}(t)\nonumber\\
	&=\lim\limits_{A_0\rightarrow 0}\frac{1}{T}\bar{\mathbf{X}}_p\bar{\mathbf{X}}_q^\mathsf{H}\in\mathbb{C}^{N\times N}, p,q\in\mathcal{P},
\end{align}
where $\bar{\mathbf{X}}_p=\left[\mathbi{x}_p(1),\mathbi{x}_p(2),...,\mathbi{x}_p(T)\right]\in \mathbb{R}^{N\times T}$. It is worth pointing out that $T\geq M$ is required to derive a rank $M$ approximation of ${\mathbf{R}}_{pq}^{\mathbf{X}}$ \cite{9384289}. Here, there exists a challenge that the dimension of  $\hat{\mathbf{R}}_{pq}^{\mathbf{X}}$ is theoretically infinite due to the infinitesimal units of CAPA, rendering the eigendecomposition of $\hat{\mathbf{R}}_{pq}^{\mathbf{X}}$  intractable.

To cope with this challenge, we design an equivalent transforming mechanism. According to the definition of eigenvalues and eigenvectors, the eigenvalue $\lambda$ and eigenvector $\mathbf{u}$ of the matrix $\hat{\mathbf{R}}_{pq}^{\mathbf{X}}$ satisfy $\hat{\mathbf{R}}_{pq}^{\mathbf{X}} \mathbf{u}_{pq}= \lambda \mathbf{u}_{pq}$. Multiplying both sides by $\bar{\mathbf{X}}^\mathsf{H}_q$ yields that:
\begin{align}
	\underbrace{\frac{1}{T}\lim\limits_{A_0\rightarrow 0}\bar{\mathbf{X}}_q^\mathsf{H}\bar{\mathbf{X}}_p}_{\mathbf{K}_{pq}}\bar{\mathbf{X}}_q^\mathsf{H}\mathbf{u}_{pq}= \lambda\underbrace{\lim\limits_{A_0\rightarrow 0}\bar{\mathbf{X}}_q^\mathsf{H}\mathbf{u}_{pq}}_{\mathbf{u}_p^\prime},
\end{align}
where $\mathbf{K}_{pq}=\frac{1}{T}\lim\limits_{A_0\rightarrow 0}\bar{\mathbf{X}}_q^\mathsf{H}\bar{\mathbf{X}}_p\in\mathbb{C}^{T\times N}$ and $\mathbf{u}_{pq}^\prime=\lim\limits_{A_0\rightarrow0}\bar{\mathbf{X}}_p^\mathsf{H}\mathbf{u}_{pq}\in\mathbb{C}^{T}$ are defined as the equivalent covariance matrix and the equivalent eigenvector, respectively. Then, noticing that matrices  $\mathbf{K}_{pq}$ and $\hat{\mathbf{R}}_{pq}^\mathbf{X}$ share a number of $M$ identical vectors, the eigendecomposition of the infinite matrix $\hat{\mathbf{R}}_{pq}^\mathbf{X}$ can be equivalently reduced to finding the eigenvectors of a finite matrix $\mathbf{K}_{pq}$ and recovering $\mathbf{u}_{pq}$ from $\mathbf{u}^\prime_{pq}$.

Subsequently, we will focus on the calculation of $\mathbf{K}_{pq}$ and $\mathbf{u}^\prime_{pq}$. The $(i,j)$-th entry of $\mathbf{K}_{pq}$ can be calculated as
\begin{align}\label{eq:::29}
	[\mathbf{K}_{pq}]_{i,j}=&\lim\limits_{A_0\rightarrow 0}\frac{1}{T}\sum_{n=1}^{N}[\mathbf{x}^{\ast}_n(i)]_q[\mathbf{x}_n(j)]_p\nonumber\\
	=&\frac{1}{T}\lim\limits_{A_0\rightarrow 0}A_0^2\sum_{n=1}^{N}e_q^\ast(\mathbf{r}_n,i)e_p(\mathbf{r}_n,j)\nonumber\\
	=&\frac{1}{T}\int_{\mathcal{S}}e_q^\ast(\mathbf{r},i)e_p(\mathbf{r},j)d\mathbf{r}, i,j\in\mathcal{T}.
\end{align}
The calculation of integral in Eq. \eqref{eq:::29} is incurs both high time and space complexity, potentially limiting the algorithm feasibility. In this paper, we involve the Gauss-Legendre quadrature to promote the calculation of integration. The integral of a function $f(x,y)$ over region $[a,b]$ can approximated by the following Gauss-Legendre quadrature \cite{swarztrauber2003computing}:
\begin{align}
	\int_{a}^{b}f(x)dx\approx\frac{b-a}{2} \sum^{K}_{k=1}\omega_kf(\frac{b-a}{2}\theta_k+\frac{a+b}{2}),
\end{align}
where $K\in\mathbb{Z}^+$ is the order of Gauss-Legendre quadrature, equaling to the number of sampling points, and typically a larger $K$ contributes to higher approximation accuracy but also increased computational burden. $\theta_k$ and $\omega_k$ denote the roots of Gauss-Legendre polynomial and the corresponding weight, respectively, which are formulated as follows:
\begin{align}\label{eq:31}
	\omega_k=\frac{2}{(1-\theta_k^2)[P_K^\prime(\theta_k)]^2}, k\in \mathcal{K},
\end{align}  
where $\mathcal{K}=\{1,2,...,K\}$, $P_K^\prime(\theta_k)$ is the first-order differentials of the $K$-th Gauss-Legendre polynomial at $\theta_k$. Based on this technique, Eq. \eqref{eq:::29} can be further reformulatd as
\begin{align}\label{eq:::32}
	&[\mathbf{K}_{pq}]_{i,j}\nonumber\\
	=&\frac{1}{T}\int_{-\frac{L_y}{2}}^{\frac{L_y}{2}}\int_{-\frac{L_x}{2}}^{\frac{L_x}{2}}e_q^\ast(\mathbf{r}_n,i)e_p(\mathbf{r}_n,j)dr_xdr_y\nonumber\\
	\approx& \frac{L_xL_y}{4T}\sum_{k_x=1}^{K}\sum_{k_y=1}^{K}\omega_{k_x}\omega_{k_y}e_q^\ast(\tilde{\mathbf{r}}_{k_x,k_y},i)e_p(\tilde{\mathbf{r}}_{k_x,k_y},j),
\end{align}
where $\tilde{\mathbf{r}}_{k_x,k_y}=\left[\theta_{k_x}L_x/2,\theta_{k_y}L_y/2,0\right]^\mathsf{T}$. Notably, the dimension of $\mathbf{u}$ is the same with the number of receive units in CAPA, such that each element can be treated as a variable dependant on the corresponding array unit: $\mathbf{u}=[u(\mathbf{r}_1),u(\mathbf{r}_2),...,u(\mathbf{r}_N)]^\mathsf{T}$. Then, similar to Eq. \eqref{eq:::32},  $\mathbf{u}_{pq}^\prime=\lim\limits_{A_0\rightarrow 0}\bar{\mathbf{X}}_q^\mathsf{H}\mathbf{u}$ can be rewritten into
\begin{align}\label{eq:ui}
	\mathbf{u}_{pq}^\prime=&\int_{-\frac{L_y}{2}}^{\frac{L_y}{2}}\int_{-\frac{L_x}{2}}^{\frac{L_x}{2}}\bar{\mathbf{e}}_{q}(\mathbf{r})u(\mathbf{r})dr_xdr_y \nonumber\\
	\approx & \frac{L_xL_y}{4}\sum_{k_x=1}^{K}\sum_{k_y=1}^{K}\omega_{k_x}\omega_{k_y}\bar{\mathbf{e}}_q(\tilde{\mathbf{r}}_{k_x,k_y})u(\tilde{\mathbf{r}}_{k_x,k_y})\nonumber\\
	=& \bar{\mathbf{E}}_q\boldsymbol{\Omega}\bar{\mathbf{u}}_{pq},
\end{align}
where $\bar{\mathbf{e}}_q(\mathbf{r})=[e_q(\mathbf{r},1),e_q(\mathbf{r},2),...,e_q(\mathbf{r},T)]^\mathsf{T}\in\mathbb{C}^{T}$, and 
\begin{align}
	&\bar{\mathbf{E}}_q=[\bar{\mathbf{e}}_q(\tilde{\mathbf{r}}_{1,1}),\bar{\mathbf{e}}_q(\tilde{\mathbf{r}}_{1,2}),...,\bar{\mathbf{e}}_q(\tilde{\mathbf{r}}_{K,K})]\in\mathbb{C}^{T\times K^2},\label{eq:34}\\
	&\boldsymbol{\Omega} = \operatorname{diag}\{\omega_1\omega_1,\omega_1\omega_2,...,\omega_K\omega_K\}\in\mathbb{R}^{K^2\times K^2},\label{eq:35}\\
	&\bar{\mathbf{u}}_{pq}=[u(\tilde{\mathbf{r}}_{1,1}),u(\tilde{\mathbf{r}}_{1,2}),...,u(\tilde{\mathbf{r}}_{K,K})]^\mathsf{H}\in\mathbb{C}^{K^2}.
\end{align}

Notably, there are a total of $T$ eigenvectors for $\mathbf{K}_{pq}$, denoted by $\mathbf{u}_{pq,1}^\prime,\mathbf{u}_{pq,2}^\prime,...,\mathbf{u}_{pq,T}^\prime$.
Subsequenlty, the eigenvectors of the noise subspace can be recovered from $\mathbf{u}^\prime_{pq,i}$ as follows:
\begin{align}\label{eq:37}
	\bar{\mathbf{U}}_{pq,2}=&(\bar{\mathbf{E}}\boldsymbol{\Omega})^\dag[\mathbf{u}^\prime_{pq,M+1},\mathbf{u}^\prime_{pq,M+2},...,\mathbf{u}^\prime_{pq,T}]\nonumber\\
	&\in\mathbb{C}^{K^2\times (T-M)}.
\end{align}

Finally, the DOA parameters $\{\hat{\theta}_m,\hat{\phi}_m\}_{m=1}^M$ can be obtained via searching for the $M$ peaks in the following MUSIC spectrum for tri-polarized CAPA:
\begin{align}\label{eq:39}
	&\{\hat{\theta}_m,\hat{\phi}_m\}_{m=1}^M\nonumber\\
	=&{}^M\arg\max_{\theta,\phi}\prod_{p,q\in\mathcal{P}}\left(1+\lim\limits_{A_0\rightarrow 0}\frac{1}{\| \boldsymbol{\alpha}^\mathsf{H}(\theta,\phi)\mathbf{U}_{pq,2}\|_2}\right)\nonumber\\
	\approx& {}^M\arg\max_{\theta,\phi}\prod_{p,q\in\mathcal{P}}\left(1+\frac{1}{\| \bar{\boldsymbol{\alpha}}^\mathsf{H}(\theta,\phi)\bar{\mathbf{U}}_{pq,2}\|_2}\right),
\end{align}
where ${}^M\arg\max_{\theta,\phi}f(\cdot)$ is defined as the set of $M$ pairs of  $\{\theta,\phi\}$ at which $f(\cdot)$ reaches its local maxima, and 
\begin{align}
	&\boldsymbol{\alpha}(\theta,\phi)=[a(\mathbf{r}_1,{\theta},{\phi}),a(\mathbf{r}_2,{\theta},{\phi}),...,a(\mathbf{r}_N,{\theta},{\phi})]^\mathsf{H},\\
	&\bar{\boldsymbol{\alpha}}(\theta,\phi)=[a(\tilde{\mathbf{r}}_{1,1},{\theta},{\phi}),a(\tilde{\mathbf{r}}_{1,2},{\theta},{\phi}),...,a(\tilde{\mathbf{r}}_{K,K},{\theta},{\phi})]^\mathsf{H}.
\end{align}
The addition of 1 in the spectrum function serves to avoid cases where certain polarization directions render the spectrum value zero, which would otherwise cause the product output to be zero and prevent the detection of actual DOA peaks.

\begin{remark}
	\normalfont Conventional antennas are limited to receiving signals along a single polarization axis, which may result in peak loss in the spectrum when the target’s dipole orientation aligns unfavorably. In contrast, the proposed tri-polarized receiving antenna exploits both the self- and cross-covariance matrices and developing a novel tri-polarized spectrum, realizing reliable DOA estimation as long as $\mathbf{v}_m \nparallel \mathbf{z}_m$. This enables more robust DOA acquisition under diverse target orientations.
\end{remark}

The main procedures of the aforementioned algorithm are summarized in Algorithm \ref{alg:CAPA-MUSIC}.

\begin{algorithm}[t]
	\caption{CAPA-MUSIC Algorithm}
	\begin{algorithmic}[1]\label{alg:CAPA-MUSIC}
		\REQUIRE CAPA aperture parameters (area $L_x\times L_y$ and region $\mathcal{S}$;
		Number of snapshots $T$ and the received signal $\mathbf{X}(t)$ for $t\in\mathcal{T}$;
		Legendre polynomial dimension $K$.
		
		\ENSURE Estimated target azimuth and elevation angles $\{\hat{\theta}_m,\hat{\phi}_m\}_{m=1}^M$.
		\STATE Compute the roots $\theta_k$ and weights $\omega_k$ for $K$-th Legendre polynomial using Eq. \eqref{eq:31}.
		Compute $\boldsymbol{\Omega}$ using Eq. (\ref{eq:34})
		\FOR{$p\in\{x,y,z\}$}
		\STATE Compute $\bar{\mathbf{E}}_p$ using Eq. \eqref{eq:35}.
		\FOR{$q\in\{x,y,z\}$}
		\STATE Compute matrix $\mathbf{K}_{pq}$ by using Eq. (\ref{eq:::32}).
		\STATE Compute eigenvectors $\mathbf{u}^\prime_{pq,1},\mathbf{u}^\prime_{pq,2},...,\mathbf{u}^\prime_{pq,T}$ using eigendecomposition.
		\STATE Compute $\bar{\mathbf{U}}_{pq,2}$  using Eq.   (\ref{eq:37}).
		\ENDFOR
		\ENDFOR		
		\STATE Compute the MUSIC spectrum using Eq. (\ref{eq:39}).
		\STATE Search on the MUSIC spectrum to find the local maximas $\{\hat{\theta}_m,\hat{\phi}_m\}_{m=1}^M$.
		\RETURN Estimated target azimuth and elevation angles $\{\hat{\theta}_m,\hat{\phi}_m\}_{m=1}^M$.
	\end{algorithmic}
\end{algorithm}

\subsection{Theoretical Analyses for Attitude Estimation}\label{III.C}
To better understand the identifiability of the attitude vector, we revisit the received data model in Eq.~\eqref{eq:::42}:
\begin{align}\label{eq:::42}
	\mathbf{X}(t)=\lim_{A_0\rightarrow 0}A_0 \left[\mathbf{A}(\boldsymbol{\theta},\boldsymbol{\phi})\mathbf{S}(t)\mathbf{V}+\mathbf{N}(t)\right],
\end{align}
where $\mathbf{A}(\boldsymbol{\theta},\boldsymbol{\phi})$ has been estimated in Sec. \ref{III.A}. We aim to estimate $\mathbf{V}$ from Eq. \eqref{eq:::42} in the rest of this section. Notably, there are two unknown variables in Eq. \eqref{eq:::42}: $\mathbf{S}(t)$ and $\mathbf{V}$, which formulate a typical bilinear map satisfying that:
\begin{align}
	&\forall \boldsymbol{\Psi}=\operatorname{diag}\{\psi_1,\psi_2,...,\psi_M\}>0, \nonumber\\
	&\mathbf{S}(t)\mathbf{V}=\left(\mathbf{S}(t)\boldsymbol{\Psi}\right)\left(\boldsymbol{\Psi}^{-1}\mathbf{V}\right),
\end{align}
indicating that Eq. \eqref{eq:::42} has infinite possible solutions, which is the inherent scale ambiguity in bilinear map. Hence, by resorting to the scale ambiguity theory \cite{choudhary2014identifiability}, it is intractable to estimate $\mathbf{V}$ individually.

To circumvent this ambiguity, we introduce a composite representation:
\begin{align}
	\boldsymbol{\Gamma}(t)=\mathbf{S}(t)\mathbf{V}.
\end{align}
This allows us to treat the recovery of $\boldsymbol{\Gamma}(t)$ as a standard linear estimation problem from Eq.~\eqref{eq:::42}, and subsequently focus on extracting structural information about $\mathbf{V}$ from $\boldsymbol{\Gamma}(t)$. Since the estimation of  $\boldsymbol{\Gamma}(t)$ from Eq. \eqref{eq:::42} reduces to a standard linear estimation problem, we mainly focus on extracting $\mathbf{V}$ from $\boldsymbol{\Gamma}(t)$. Specifically, $\boldsymbol{\Gamma}(t)$ can be formulated as follows:
\begin{align}
	\boldsymbol{\Gamma}(t)=\begin{bmatrix}
		s_1(t) & 0 & \cdots & 0 \\
		0 & s_2(t) & \cdots & 0 \\
		\vdots & \vdots & \ddots & \vdots \\
		0 & 0 & \cdots & s_M(t)
	\end{bmatrix}
	\begin{bmatrix}
		\mathbf{v}_1^\mathsf{T}\\
		\mathbf{v}_2^\mathsf{T}\\
		\vdots\\
		\mathbf{v}_M^\mathsf{T}\\
	\end{bmatrix},
\end{align}
where each row is $\mathbf{v}_m=(\mathbf{I}-\bar{\mathbf{z}}_m\bar{\mathbf{z}}_m^\mathsf{T})\bar{\mathbf{q}}_m$.

Notably, the attitude vector $\bar{\mathbf{q}}_m$ can be decomposed with respect to $\bar{\mathbf{z}}_m$ as:
\begin{align}
	\bar{\mathbf{q}}_m=\bar{\mathbf{q}}_m^{\perp \mathbf{z}}+\bar{\mathbf{q}}_m^{\|\mathbf{z}}, m\in\mathcal{M},
\end{align}
where
\begin{align}
	&\bar{\mathbf{q}}_m^{\perp \mathbf{z}}=\bar{\mathbf{q}}_m-\bar{\mathbf{z}}_m^\mathsf{T}\bar{\mathbf{q}}_m\bar{\mathbf{z}}_m,\\
	&\bar{\mathbf{q}}_m^{\parallel\mathbf{z}}=\bar{\mathbf{z}}_m^\mathsf{T}\bar{\mathbf{q}}_m\bar{\mathbf{z}}_m,
\end{align}
represent the components perpendicular and parallel to vector $\bar{\mathbf{z}}_m$, respectively. The projection operator $\mathbf{I}-\bar{\mathbf{z}}_m\bar{\mathbf{z}}_m^\mathsf{T}$ inherently nulls the components of $\bar{\mathbf{q}}_m$ parallel to 
$\mathbf{z}_m$, thereby retaining only the orthogonal part, which can be formulated as follows:
\begin{align}
	&(\mathbf{I}-\bar{\mathbf{z}}_m\bar{\mathbf{z}}_m^\mathsf{T})\bar{\mathbf{q}}_m\nonumber\\
	=&(\mathbf{I}-\bar{\mathbf{z}}_m\bar{\mathbf{z}}_m^\mathsf{T})\bar{\mathbf{q}}_m^{\perp \mathbf{z}}+\underbrace{(\mathbf{I}-\bar{\mathbf{z}}_m\bar{\mathbf{z}}_m^\mathsf{T})\bar{\mathbf{q}}_m^{\parallel\mathbf{z}}}_{=0}.
\end{align}

\begin{figure}
	\centering
	\includegraphics[width=0.5\linewidth]{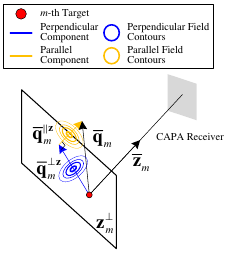}
	\caption{Illustration of the identifiability of attitude vectors}
	\label{fig:altitudeillustration}
\end{figure}

Therefore, it can be concluded that the  parallel components of $\bar{\mathbf{q}}_m$ is neglected in $\boldsymbol{\Gamma}(t)$ due to the existence of orthogonal projection matrix $\mathbf{I}-\bar{\mathbf{z}}_m\bar{\mathbf{z}}^\mathsf{T}_m$, such that only the perpendicular components $\bar{\mathbf{q}}_m^{\perp \mathbf{z}}$ can be estimated from $\boldsymbol{\Gamma}(t)$ when the signal snapshot $\mathbf{S}(t)$ are unknown. Illustratively, we depict the components of the attitude vector in Fig.  \ref{fig:altitudeillustration}. Due to the transverse nature of EM wave, it can be observed that only the perpendicular components of source signals can be propagated to the CAPA receiver. Hence, we give the following proposition to summarize the observations:
\begin{proposition}[Identifiability of the attitude vector]
	\normalfont For the utilized tri-polarization signal model in Eq. \eqref{eq:20}, given that the snapshot signal $s_m(t)$ is unknown, the full attitude vector $\bar{\mathbf{q}}_m$ is not uniquely identifiable when $m\in\mathcal{M}$ is unknown.  However, the direction of the perpendicular components $\bar{\mathbf{q}}_m^{\perp \mathbf{z}}$ remains identifiable.
\end{proposition}

Based on this observation, we proceed the attitude estimation algorithm design in two situations, i.e., the general case with unknown $s_m(t)$ and the special case with known $s_m(t)$.

\subsection{Attitude Estimation Algorithm Design}\label{III.D}
In this section, we aim to estimate the attitude parameters $\{\bar{\mathbf{q}}_m\}_{m=1}^{M}$ given the received signal $\mathbf{X}$ and the estimated DOA parameters $\{\hat{\theta}_m,\hat{\phi}_m\}_{m=1}^M$
\subsubsection{General Cases without Target Information}
In the case, $s_m(t)$ is unknown and only the perpendicular components $\bar{\mathbf{q}}_m^{\perp \mathbf{z}}$ can be estimated. At time $t$,  the received signal $\mathbf{X}(t)$ is formulated in Eq. \eqref{eq:20}. By defining $\boldsymbol{\Gamma}(t)=\mathbf{S}(t)\mathbf{V}$, Eq. \eqref{eq:20} can be rewritten into the following form:
\begin{align}\label{eq:42}
	\mathbf{X}(t)=\lim\limits_{A_0\rightarrow 0}A_0[\mathbf{A}(\boldsymbol{\theta},\boldsymbol{\phi}) \boldsymbol{\Gamma}(t)+\mathbf{N}(t)].
\end{align}
Apparently, $\boldsymbol{\Gamma}(t)$ can be estimated from Eq. \eqref{eq:42} by using the following least square model:
\begin{align}
	\hat{\boldsymbol{\Gamma}}(t)=\arg\min_{\boldsymbol{\Gamma}(t)}\lim\limits_{A_0\rightarrow 0}\left\|\mathbf{X}(t)-A_0\mathbf{A}(\boldsymbol{\theta},\boldsymbol{\phi})\boldsymbol{\Gamma}(t)\right\|_2^2.
\end{align}
By setting the derivative with regard to $\boldsymbol{\Gamma}(t)$ to zero, one has
\begin{align}\label{eq:44}
	\lim\limits_{A_0\rightarrow 0}A_0\mathbf{A}^\mathsf{H}(\boldsymbol{\theta},\boldsymbol{\phi})\mathbf{X}(t)=\lim\limits_{A_0\rightarrow 0}A_0^2\mathbf{A}^\mathsf{H}(\boldsymbol{\theta},\boldsymbol{\phi})\mathbf{A}(\boldsymbol{\theta},\boldsymbol{\phi})\boldsymbol{\Gamma}(t).
\end{align} 
Similar to the derivation of $\mathbf{K}_{pq}$, the terms on the left side of Eq. \eqref{eq:44} can be reformulated as follows:
\begin{align}
	&\lim\limits_{A_0\rightarrow 0}A_0\mathbf{A}^\mathsf{H}(\boldsymbol{\theta},\boldsymbol{\phi})\mathbf{X}(t)\nonumber\\
	=&\int_{\mathcal{S}}\mathbf{a}^\mathsf{H}(\mathbf{r},\boldsymbol{\theta},\boldsymbol{\phi})\mathbf{e}(\mathbf{r},t)d\mathbf{r}\triangleq\mathbf{Q}(\boldsymbol{\theta},\boldsymbol{\phi},t)\in\mathbb{C}^{M\times 3}.
\end{align}
For $i\in\mathcal{M}$ and $j\in\{1,2,3\}$, the $(i,j)$-th entry of $\mathbf{Q}(\boldsymbol{\theta},\boldsymbol{\phi},t)$ is given by:
\begin{align}\label{eq:46}
	&\left[\mathbf{Q}(\boldsymbol{\theta},\boldsymbol{\phi},t)\right]_{i,j}\nonumber\\
	=&\int_{-\frac{L_y}{2}}^{\frac{L_y}{2}}\int_{-\frac{L_x}{2}}^{\frac{L_x}{2}} a^\ast(\mathbf{r},{\theta}_i,{\phi}_i)[\mathbf{e}(\mathbf{r},t)]_{j} dr_xdr_y\nonumber\\
	=&\int_{-\frac{L_y}{2}}^{\frac{L_y}{2}}\int_{-\frac{L_x}{2}}^{\frac{L_x}{2}}e^{\mathsf{j}k\mathbf{r}^\mathsf{T}\mathbf{d}(\theta_i,\phi_i)}[\mathbf{e}(\mathbf{r},t)]_{j} dr_xdr_y.
\end{align}
In practical implementations, the evaluation of Eq.~\eqref{eq:46} via direct numerical integration can be computationally prohibitive, since each candidate direction $(\theta_i,\phi_i)$ requires a separate two-dimensional integration over the aperture. By contrast, the Fourier transform provides a natural and efficient framework: the integral in Eq.~\eqref{eq:46} corresponds to sampling the two-dimensional spatial Fourier transform of the received field at given frequency points. Specifically, we have the well-known 2-D Fourier transform:
\begin{align}
	\mathcal{F}_{\operatorname{2D}}\{f(x, y)\}|(f_x, f_y)
	= \iint_{-\infty}^{\infty} f(x, y) \, e^{-\mathsf{j}2\pi (f_x x + f_y y)} dxdy.
\end{align}
Hence, we find out that Eq. \eqref{eq:46} is a 2D windowed Fourier transform of received signal $ [\mathbf{e}(\mathbf{r},t)]_j$, which can be simplified as follows:
\begin{align}\label{eq:48}
	&\left[\mathbf{Q}(\boldsymbol{\theta},\boldsymbol{\phi},t)\right]_{i,j}\nonumber\\
	= &\iint_{-\infty}^{\infty} [\mathbf{e}(\mathbf{r},t)]_j  \Pi\left( \frac{r_x}{L_x}, \frac{r_y}{L_y} \right) 
	e^{\mathsf{j}k(d_{x,i} r_x + d_{y,i} r_y)} \, dr_x dr_y \nonumber\\
	= &\left.\mathcal{F}_{\operatorname{2D}} \left\{ [\mathbf{e}(\mathbf{r},t)]_j  \Pi\left( \frac{r_x}{L_x}, \frac{r_y}{L_y} \right) \right\}\right|{\left( -\frac{k}{2\pi} d_{x,i}, -\frac{k}{2\pi} d_{y,i} \right)}
	,
\end{align}
where $d_{x,i}=\cos\theta_i\cos\phi_i$, $d_{y,i}=\sin\theta_i\cos\phi_i$, and $\Pi(\cdot,\cdot)$ denotes the 2D rectangular window. Leveraging the fast fourier transform (FFT), the entire spectrum can be computed with remarkably reduced complexity. This enables efficient evaluation of the response over multiple targets, which is particularly advantageous for real-time and high-resolution attitude estimation.

Subsequently, we focus on the term in the right side of Eq. \eqref{eq:44}. Defining that $\boldsymbol{\Xi}\triangleq\lim\limits_{A_0\rightarrow 0}A_0^2\mathbf{A}^\mathsf{H}(\boldsymbol{\theta},\boldsymbol{\phi})\mathbf{A}(\boldsymbol{\theta},\boldsymbol{\phi})\in\mathbb{C}^{M\times M}$ and the $(i,j)$-th entry of $\boldsymbol{\Xi}$ can be formulated as follows:
\begin{align}\label{eq:::54xi}
	[\boldsymbol{\Xi}]_{i,j}=&\lim\limits_{A_0\rightarrow 0}A_0^2\boldsymbol{\alpha}^\mathsf{H}(\theta_i,\phi_i)\boldsymbol{\alpha}(\theta_j,\phi_j)\nonumber\\
	=&\int_{\mathcal{S}}a^\ast(\mathbf{r},\theta_i,\phi_i)a(\mathbf{r},\theta_j,\phi_j)d\mathbf{r}\nonumber\\
	=&\int_{-\frac{L_y}{2}}^{\frac{L_y}{2}}\int_{-\frac{L_x}{2}}^{\frac{L_x}{2}}  e^{\mathsf{j}k\mathbf{r}^\mathsf{T}[\mathbf{d}(\theta_i,\phi_i)-\mathbf{d}(\theta_j,\phi_j)]} dr_x dr_y\nonumber\\
	=&L_x L_y
	\operatorname{sinc}\left( \frac{k\,\Delta d_{i,j}^{(x)}\,L_x}{2} \right)
	\operatorname{sinc}\left( \frac{k\,\Delta d_{i,j}^{(y)}\,L_y}{2} \right), 
\end{align}
where
\begin{align}
	&\Delta d_{i,j}^{(x)} = \cos\theta_i \cos\phi_i - \cos\theta_j \cos\phi_j,i,j\in \mathcal{M}.\nonumber\\
	&\Delta d_{i,j}^{(y)} = \sin\theta_i \cos\phi_i - \sin\theta_j \cos\phi_j, i,j\in \mathcal{M}.\nonumber
\end{align}
Then, the right side of Eq. \eqref{eq:44} equals to $\boldsymbol{\Xi}\boldsymbol{\Gamma}(t)$
and $\boldsymbol{\Gamma}(t)$ can be estimated by
\begin{equation}\label{eq:computga}
	\hat{\boldsymbol{\Gamma}}(t)=\boldsymbol{\Xi}^{-1}\mathbf{Q}(\boldsymbol{\theta},\boldsymbol{\phi},t),
\end{equation}
with the elements of  $\mathbf{Q}(\boldsymbol{\theta},\boldsymbol{\phi},t)$ and $\boldsymbol{\Xi}$ elaborated in Eqs. \eqref{eq:48} and \eqref{eq:::54xi}, respectively.

Then, we recollect the definition of $\boldsymbol{\Gamma}(t)=\mathbf{S}(t)\mathbf{V}$. By defining the $m$-th row of $\boldsymbol{\Gamma}(t)$ as $\boldsymbol{\gamma}_m(t)$, one has:
\begin{align}\label{eq::56}
	\boldsymbol{\gamma}_m(t)=&s_m(t)(\mathbf{I}-\bar{\mathbf{z}}_m\bar{\mathbf{z}}_m^\mathsf{T})\bar{\mathbf{q}}_m\nonumber\\
	=&s_m(t)(\mathbf{I}-\bar{\mathbf{z}}_m\bar{\mathbf{z}}_m^\mathsf{T})(\bar{\mathbf{q}}_m^{\perp \mathbf{z}}+\bar{\mathbf{q}}_m^{\|\mathbf{z}}).
\end{align}
Due to the fact that $\mathbf{I}-\bar{\mathbf{z}}_m\bar{\mathbf{z}}_m^\mathsf{T}$ is an orthogonal projection matrix and the fact that $(\mathbf{I}-\bar{\mathbf{z}}_m\bar{\mathbf{z}}_m^\mathsf{T})\bar{\mathbf{q}}_m^{\|\mathbf{z}}=\mathbf{0}$, multiplying both sides of Eq. \eqref{eq::56} by $\mathbf{I}-\bar{\mathbf{z}}_m\bar{\mathbf{z}}_m^\mathsf{T}$ yields that:
\begin{align}\label{eq::57}
	&(\mathbf{I}-\bar{\mathbf{z}}_m\bar{\mathbf{z}}_m^\mathsf{T})\boldsymbol{\gamma}_m(t)\nonumber\\
	=&s_m(t)(\mathbf{I}-\bar{\mathbf{z}}_m\bar{\mathbf{z}}_m^\mathsf{T})\bar{\mathbf{q}}_m^{\perp \mathbf{z}}
	\stackrel{(a)}{=}s_m(t)\bar{\mathbf{q}}_m^{\perp \mathbf{z}},
\end{align}
where $(a)$ is due to the fact that $\bar{\mathbf{z}}_m^\mathsf{T}\bar{\mathbf{q}}_m^{\perp \mathbf{z}}=0$ for $\forall m\in\mathcal{M}$. Notably, $\boldsymbol{\gamma}_m$ is real-valued, while $s_m(t)$ is complex-valued. Therefore, we can rewrite Eq. \eqref{eq::57} into:
\begin{align}
	(\mathbf{I}-\bar{\mathbf{z}}_m\bar{\mathbf{z}}_m^\mathsf{T})\Re\{\boldsymbol{\gamma}_m(t)\}=\Re\{s_m(t)\}\bar{\mathbf{q}}_m^{\perp \mathbf{z}},\\
	(\mathbf{I}-\bar{\mathbf{z}}_m\bar{\mathbf{z}}_m^\mathsf{T})\Im\{\boldsymbol{\gamma}_m(t)\}=\Im\{s_m(t)\}\bar{\mathbf{q}}_m^{\perp \mathbf{z}}.
\end{align}
By aggregating the information across $t\in\mathcal{T}$, it can be further derived that:
\begin{align}\label{eq::60}
	(\mathbf{I}-\bar{\mathbf{z}}_m\bar{\mathbf{z}}_m^\mathsf{T})\mathbf{G}_m=\bar{\mathbf{q}}_m^{\perp \mathbf{z}}\boldsymbol{\xi}_m^\mathsf{T},
\end{align}
where 
\begin{align}\label{eq:::61}
	&\mathbf{G}_m=\left[\Re\{\boldsymbol{\gamma}_m(1)\},\Im\{\boldsymbol{\gamma}_m(1)\},\Re\{\boldsymbol{\gamma}_m(2)\},\Im\{\boldsymbol{\gamma}_m(2)\},\right.\nonumber\\ 
	&\left.\quad\quad\quad...,\Re\{\boldsymbol{\gamma}_m(T)\},\Im\{\boldsymbol{\gamma}_m(T)\}\right]\in\mathbb{R}^{3\times 2T},\\
	\label{eq:::62}
	&\boldsymbol{\xi}_m=\left[\Re\{s_m(1)\},\Im\{s_m(1)\},\Re\{s_m(2)\},\Im\{s_m(2)\},\right.\nonumber\\
	&\left.\quad\quad\quad...,\Re\{s_m(T)\},\Im\{s_m(T)\}\right]^\mathsf{T}\in\mathbb{R}^{2T}.
\end{align}
In the absence of noise, the rank of Eq. \eqref{eq::60} equals to $1$. In the presence of noise, estimating the direction of $\bar{\mathbf{q}}_m^{\perp \mathbf{z}}$ becomes a least square rank-1 approximation problem. In this case, the optimal estimate of  $\bar{\mathbf{q}}_m^{\perp \mathbf{z}}$ corresponds to (up to a scaling factor) the first left singular vector of matrix $(\mathbf{I}-\bar{\mathbf{z}}_m\bar{\mathbf{z}}_m^\mathsf{T})\mathbf{G}_m$, which is formulated by:
\begin{align}\label{eq:::63}
	\exists \kappa\in\mathbb{R},~ \hat{\bar{\mathbf{q}}}_m^{\perp \mathbf{z}}=\kappa[\tilde{\mathbf{U}}]_{:,1},
\end{align}
where $\tilde{\mathbf{U}}$ is derived by the SVD decomposition $(\mathbf{I}-\bar{\mathbf{z}}_m\bar{\mathbf{z}}_m^\mathsf{T})\mathbf{G}_m=\tilde{\mathbf{U}}\boldsymbol{\Sigma}\tilde{\mathbf{V}}^\mathsf{H}$. Furthermore, the full attitude $\hat{\bar{\mathbf{q}}}_m$ should takes the form of:
\begin{align}
	\hat{\bar{\mathbf{q}}}_m=\kappa_1\tilde{\mathbf{u}}_1+\kappa_2\bar{\mathbf{z}}_m,
\end{align}
where $\kappa_1,\kappa_2\in\mathbb{R}$ satisfying $\kappa_1^2+\kappa_2^2=1$.

\subsubsection{Special Cases with Known Target Information}
In this case, $\boldsymbol{\xi}_m$ is known, and the attitude vector can be directly estimated. 
Starting from Eq.~\eqref{eq::60}, right-multiplying both sides by $\boldsymbol{\xi}_m$ yields  
\begin{align}
	(\mathbf{I}-\bar{\mathbf{z}}_m\bar{\mathbf{z}}_m^\mathsf{T})\mathbf{G}_m\boldsymbol{\xi}_m
	= \bar{\mathbf{q}}_m^{\perp \mathbf{z}}\,(\boldsymbol{\xi}_m^\mathsf{T}\boldsymbol{\xi}_m).
\end{align}
Evidently, as long as there exists $t \in \mathcal{T}$ such that $s_m(t) \neq 0$, implying that the $m$-th source current is not identically zero, $\bar{\mathbf{q}}_m^{\perp \mathbf{z}}$ can be estimated in the following closed form:
\begin{align}\label{eq:::65}
	\hat{\bar{\mathbf{q}}}_m^{\perp \mathbf{z}}
	= \frac{(\mathbf{I}-\bar{\mathbf{z}}_m\bar{\mathbf{z}}_m^\mathsf{T})\mathbf{G}_m\boldsymbol{\xi}_m}{\|\boldsymbol{\xi}_m\|_2^2}.
\end{align}
Given that $\|\bar{\mathbf{q}}_m\|_2=1$, the norm of the parallel component can be obtained as  
\begin{align}
	\|\hat{\bar{\mathbf{q}}}_m^{\parallel \mathbf{z}}\|_2
	&= \sqrt{1-\|\hat{\bar{\mathbf{q}}}_m^{\perp \mathbf{z}}\|_2^2} \nonumber\\
	&= \frac{\sqrt{\|\boldsymbol{\xi}_m\|_2^4 - \boldsymbol{\xi}_m^\mathsf{T}\mathbf{G}_m^\mathsf{T}(\mathbf{I}-\bar{\mathbf{z}}_m\bar{\mathbf{z}}_m^\mathsf{T})\mathbf{G}_m\boldsymbol{\xi}_m}}{\|\boldsymbol{\xi}_m\|_2^2}.
\end{align}
Since $\bar{\mathbf{q}}_m^{\parallel \mathbf{z}}$ is parallel to $\bar{\mathbf{z}}_m$, it follows that  
\begin{align}\label{eq:::68}
	\hat{\bar{\mathbf{q}}}_m^{\parallel \mathbf{z}} = \pm \|\hat{\bar{\mathbf{q}}}_m^{\parallel \mathbf{z}}\|_2 \cdot \bar{\mathbf{z}}_m.
\end{align}
Finally, the full attitude vector of the $m$-th source is given by  
\begin{align}\label{eq:::69}
	\hat{\bar{\mathbf{q}}}_m
	&= \hat{\bar{\mathbf{q}}}_m^{\parallel \mathbf{z}} + \hat{\bar{\mathbf{q}}}_m^{\perp \mathbf{z}} \nonumber\\
	&= \frac{1}{\|\boldsymbol{\xi}_m\|_2^2}(\mathbf{I}-\bar{\mathbf{z}}_m\bar{\mathbf{z}}_m^\mathsf{T})\mathbf{G}_m\boldsymbol{\xi}_m\pm \|\hat{\bar{\mathbf{q}}}_m^{\parallel \mathbf{z}}\|_2 \cdot \bar{\mathbf{z}}_m.
\end{align}

\begin{remark}
	\normalfont It is worth noting that an \emph{attitude ambiguity} arises in this estimation process. This ambiguity originates from the fact that the orthogonal projection matrix $\mathbf{I}-\bar{\mathbf{z}}_m\bar{\mathbf{z}}_m^\mathsf{T}$ in Eq.~\eqref{eq:6} removes the directional information along $\bar{\mathbf{z}}_m$, thereby inducing a phase ambiguity in the received signal. This phenomena fundamentally origin from the transverse nature of EM waves in free space, i.e., the electric and magnetic fields are always orthogonal to the propagation direction. Consequently, two symmetric feasible solutions exist. From a geometric perspective, if $\bar{\mathbf{q}}_m$ is a valid solution, then $-\bar{\mathbf{q}}_m$ produces the same projected component in the orthogonal subspace, leading to an inherent $\pm$ ambiguity in the attitude estimation.
\end{remark}

We summarize the main procedures of the proposed attitude estimation method in Algorithm \ref{algorithm::2}.
\begin{algorithm}[t]
	\caption{The Proposed Attitude Estimation Algorithm}
	\begin{algorithmic}[1]\label{algorithm::2}
		\REQUIRE DOA matrix $\mathbf{A}(\boldsymbol{\theta},\boldsymbol{\phi})$; received signal $\mathbf{X}(t)$; known or unknown signal snapshots $s_m(t)$ for $m\in\mathcal{M}$ and $t\in\mathcal{T}$.
		\ENSURE Attitude estimates $\hat{\bar{\mathbf{q}}}_m$.
		
		\STATE Compute $\mathbf{Q}(\boldsymbol{\theta},\boldsymbol{\phi},t)$ for $t\in\mathcal{T}$ using Eq. \eqref{eq:48}
		\STATE Compute $\boldsymbol{\Xi}$ using Eq. \eqref{eq:::54xi}
		\STATE Compute $\hat{\boldsymbol{\Gamma}}$ by using Eq. \eqref{eq:computga}
		\IF{$s_m(t)$ is unknown for $m\in\mathcal{M}$ and $t\in\mathcal{T}$}
		\FOR{$m\in\mathcal{M}$}
		\STATE Compute $\mathbf{G}_m$ using Eq. \eqref{eq:::61}
		\STATE Compute the SVD of matrix $(\mathbf{I}-\bar{\mathbf{z}}_m\bar{\mathbf{z}}_m^\mathsf{T})\mathbf{G}_m$
		\STATE Compute the estimated attitude $\hat{\bar{\mathbf{q}}}_m$ using Eq. \eqref{eq:::63}
		\ENDFOR
		\ELSIF{$s_m(t)$ is unknown for $m\in\mathcal{M}$ and $t\in\mathcal{T}$}
		\FOR{$m\in\mathcal{M}$}
		\STATE Compute $\hat{\bar{\mathbf{q}}}_m^{\perp \mathbf{z}}$ using Eq. \eqref{eq:::65}
		\STATE Compute $\hat{\bar{\mathbf{q}}}_m^{\parallel \mathbf{z}}$ using Eq. \eqref{eq:::68}
		\STATE Compute the estimated attitude $\hat{\bar{\mathbf{q}}}_m$ using Eq. \eqref{eq:::69}
		\ENDFOR
		\ENDIF
		\RETURN $\hat{\bar{\mathbf{q}}}_m$.
	\end{algorithmic}
\end{algorithm}


\subsection{Complexity Analysis}\label{III.E}
\subsubsection{DOA Estimation}
In Algorithm~\ref{alg:CAPA-MUSIC}, the calculation of $\theta_k$ and $\omega_k$ in {Step~1} requires a computational cost of $\mathcal{O}(K^2)$. 
For $p,q\in\mathcal{P}$, the formation of $\bar{\mathbf{E}}_p$ in {Step~5} incurs a cost of $\mathcal{O}(K^2T^2)$, while the eigendecomposition in {Step~6} takes $\mathcal{O}(T^3)$. 
The recovery of the noise subspace $\bar{\mathbf{U}}_{pq,2}$ in {Step~7} consumes $\mathcal{O}(TK^4)$. 
Subsequently, the derivation of the MUSIC spectrum over one search grid point has a complexity of $\mathcal{O}(9K^2(T-M))$. 
Given a total of $N_s$ scanning grid points, the overall computational complexity of the proposed DOA estimation algorithm is approximated $\mathcal{O}\left(T^3 + N_sK^2(T-M)\right)$.

In practical implementations, increasing $N_s$ enhances the DOA estimation resolution but inevitably raises the computational burden. To obtain a balance between estimation and implementation feasibility, it is recommended to perform a low-resolution coarse grid scan, and then refine the results using a gradient descent method to achieve high-resolution estimates.

\subsubsection{Attitude Estimation}
In Algorithm \ref{algorithm::2}, $\mathbf{Q}(\boldsymbol{\theta},\boldsymbol{\phi},t)$, assuming that FFT is utilized to calculate matrix $\mathbf{Q}$ with $N_f$ sampling points, Step 1 takes up a complexity of $\mathcal{O}(3TN_f\log N_f)$. The calculation of $\boldsymbol{\Xi}$ and $\hat{\boldsymbol{\Gamma}}$ in Step 2 and costs a complexity of $\mathcal{O}(M^2)$ and $\mathcal{O}(M^3+TM^3)$, respectively. The SVD in Step 7 costs a complexity of $\mathcal{O}(T^2)$. Finally, the overall complexity of attitude estimation is $\mathcal{O}(3TN_flogN_f+TM^3+T^2)$.

\section{Numerical Results} \label{sec:results}
\begin{table}[t]
	\renewcommand\arraystretch{1.2}
	\caption{Parameter settings.}
	\centering\label{running_time}
	\small
	\label{tab:pstt}
	\begin{tabular}{|p{1.6cm}|p{4cm}|p{1.85cm}|}
		\hline
		Notations&Definitions&Values\\ \hline \hline
		${L_x}$&The width of CAPA&$2$m\\ \hline
		$L_y$&The height of CAPA&$2$m\\ \hline
		$\eta_0$&The free-space impedance&$120\pi\Omega$ \\ \hline
		$\sigma^2$&The noise power& $10^{-3}\text{V}^2/\text{m}^2$\\ \hline
		$K$& The dimension of Legendre polynomial& $16$\\ \hline
		$\lambda$ & The wavelength & $0.1$m \\ \hline
		$T$& The number of snapshots & $ 500$\\ \hline
	\end{tabular}
\end{table}

In this section, we provide numerical results to validate the performance of the proposed DOA and attitude sensing algorithm. 

Unless otherwise specified, the following configurations are applied throughout the simulations. Two targets are considered, with their positions and orientations specified as follows: $\mathbf{p}_1 = [-16, -10, 50]^\mathsf{T}$, $\mathbf{p}_2 = [16, -38, 40]^\mathsf{T}$, $\bar{\mathbf{q}}_1 = [0.8, 0.6, 0]^\mathsf{T}$, and $\bar{\mathbf{q}}_2 = [-0.1, 0.7, 0.7071]^\mathsf{T}$, respectively. The simulation parameters are summarized in Table~\ref{running_time}. All simulations are executed on a PC equipped with an Intel i7-13980HX 2.2GHz CPU and 32GB of RAM. The algorithms are implemented in MATLAB R2023b.

\subsection{Benchmarks and Evaluation Metrics}

The following benchmarks are considered in the simulation:

\emph{1) SPDA:} In this benchmark, an SPDA of identical size to the CAPA counterparts is utilized to receive the signals. The corresponding signal model follows the work \cite{10158997}. Notably, the SPDA is composed of discrete antennas with an effective aperture $A_d=\frac{\lambda^2}{4\pi}$ and spacing $l_d=\frac{\lambda}{2}$. Hence, there is a total of $N_d=N_{d,x}N_{d,y}$ antennas for the SPDA with $N_{d,x}=\lceil\frac{L_x}{l_d} \rceil$ and $N_{d,y}=\lceil\frac{L_y}{l_d} \rceil$, and the position of the $(m,n)$-th antenna is given by:
\begin{align}
	\bar{\mathbf{p}}_{m,n}=\left[(m-1)l_d-\frac{L_x}{2},(n-1)l_d-\frac{L_y}{2},0\right]^\mathsf{T}.
\end{align}
Then, the detailed DOA estimation algorithm can be referred to conventional MUSIC algorithms \cite{stoica2002music}.

\emph{2) Single-polarized CAPA:} Without loss of generalization, it is assumed that only the electronic field along X-axis is available, denoted as $e(\mathbf{r})=\mathbf{u}_x^\mathsf{T}\mathbf{e}(\mathbf{r})$ with $\mathbf{u}_x=[1,0,0]^\mathsf{T}$ being the unit directional vector along X-axis. Accordingly, the received signal under single polarization can be expressed as
\begin{align}
	\mathbf{x}(t)=\lim_{A_0\rightarrow 0}A_0 \left[\mathbf{A}(\boldsymbol{\theta},\boldsymbol{\phi})\mathbf{s}(t)+\mathbf{n}(t)\right]\in\mathbb{C}^{N},
\end{align}
where $\mathbf{s}(t)=[s_1(t),s_2(t),...,s_M(t)]^\mathsf{T}$ represents the source signal vector and $\mathbf{n}(t)=[n(\mathbf{r}_1,t),n(\mathbf{r}_2,t),...,n(\mathbf{r}_N,t)]^\mathbf{T}$ denotes the noise vector. Further details of the algorithm design can be found in \cite{si2025doa}.

\emph{3) CRLB:} To highlight the performance boundary of the studied CAPA system, we consider CRLB as another benchmark. In our sensing problem, the parameter vector is defined as $\boldsymbol{\beta}=[\boldsymbol{\theta};\boldsymbol{\phi};\bar{\mathbf{q}}_1;\bar{\mathbf{q}}_2;...;\bar{\mathbf{q}}_M]\in\mathbb{R}^{5M}$, the fisher information matrix (FIM) $\mathbf{J}_{\boldsymbol{\beta}}\in\mathbb{C}^{}$ is defined as follows:
\begin{align}
	\mathbf{J}_{\boldsymbol{\beta}}=\mathbb{E}\left[\left(\frac{\partial}{\partial \boldsymbol{\beta}} \mathcal{L}(\boldsymbol{\beta})\right)\left(\frac{\partial}{\partial \boldsymbol{\beta}}\mathcal{L}(\boldsymbol{\beta})\right)^\mathsf{H}\right],
\end{align}
where  $\mathcal{L}(\boldsymbol{\beta})=\ln p(\mathbf{X}(1),\mathbf{X}(2),...,\mathbf{X}(T);\boldsymbol{\beta})$ is the likelihood function and 
\begin{align}
	&p(\mathbf{X}(1),\mathbf{X}(2),...,\mathbf{X}(T);\boldsymbol{\beta})\nonumber\\
	=&\prod_{t=1}^{T}\frac{1}{(\pi\sigma^{2})^{3N}}
	\exp\!\left(
	-\frac{1}{\sigma^{2}}\big\|\mathbf{X}(t)-\mathbf{A}(\boldsymbol{\theta},\boldsymbol{\phi})\,\mathbf{S}(t)\,\mathbf{V}\big\|_\mathsf{F}^{2}
	\right).
\end{align}
Subsequently, the CRLB for DOA estimation is given by
\begin{align}
	&\mathrm{CRLB}(\theta_m)=\left[\mathbf{J}_{\boldsymbol{\beta}}^{-1}\right]_{m,m},~m\in\mathcal{M},\\
	&\mathrm{CRLB}(\phi_m)=\left[\mathbf{J}_{\boldsymbol{\beta}}^{-1}\right]_{m+M,m+M},~m\in\mathcal{M}.
\end{align}

To assess the performance of the proposed joint DOA and attitude sensing algorithm, we adopt the mean squared error (MSE) for DOA estimation and the mean angular error (MAE) for attitude estimation. Specifically, given the ground-truth azimuth angles $\boldsymbol{\theta}$ and elevation angles $\boldsymbol{\phi}$, the MSE is defined as
\begin{align}
	\mathrm{MSE}(\boldsymbol{\theta}) &= \frac{1}{N_t}\sum_{i=1}^{N_t}\|\boldsymbol{\theta}-\hat{\boldsymbol{\theta}}_i\|_2^2, \\
	\mathrm{MSE}(\boldsymbol{\phi}) &= \frac{1}{N_t}\sum_{i=1}^{N_t}\|\boldsymbol{\phi}-\hat{\boldsymbol{\phi}}_i\|_2^2,
\end{align}
where $N_t$ denotes the total number of Monte Carlo trials, and $\hat{\boldsymbol{\theta}}_i$ and $\hat{\boldsymbol{\phi}}_i$ represent the estimated azimuth and elevation angles in the $i$-th trial, respectively.  

The MAE of the attitude vector $\mathbf{z}$ is defined as
\begin{align}
	\mathrm{MAE}(\mathbf{z}) = \frac{1}{MN_t}\sum_{i=1}^{N_t}\sum_{m=1}^M\arccos\left(\frac{\hat{\mathbf{z}}_{m,i}^\mathsf{T}\mathbf{z}_m}{\|\hat{\mathbf{z}}_{m,i}\|_2\|\mathbf{z}_m\|_2}\right),
\end{align}
where $\hat{\mathbf{z}}_{m,i}$ denotes the estimated attitude vector for the $m$-th target in the $i$-th trial and $\mathbf{z}_m$ is the ground-truth attitude vector.

\subsection{Convergence of the Proposed Algorithm}
\begin{figure}
	\centering
	\includegraphics[width=0.7\linewidth]{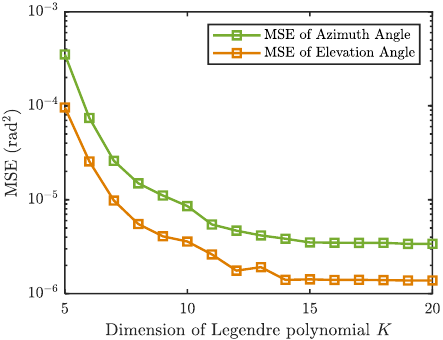}
	\caption{Convergence of the proposed algorithm}
	\label{fig:convergence}
\end{figure}

Fig.~\ref{fig:convergence} illustrates the convergence behavior of the proposed DOA sensing algorithm with respect to the dimension of the Legendre polynomial $K$, where $K$ varies within the interval $[5,20]$. In general, the MSE of DOA estimation decreases as $K$ increases, owing to the improved approximation accuracy of the integral operations involved in the algorithm design. However, once $K$ exceeds 15, the MSEs of both azimuth and elevation angles tend to stabilize at constant values, thereby confirming the convergence of the proposed algorithm. These results demonstrate the effectiveness of our design in achieving a favorable trade-off between performance and computational complexity.

\subsection{DOA Estimation Performance}

\subsubsection{MUSIC Spectrum of Different Methods}

\begin{figure*}
	\centering
	\includegraphics[width=1\linewidth]{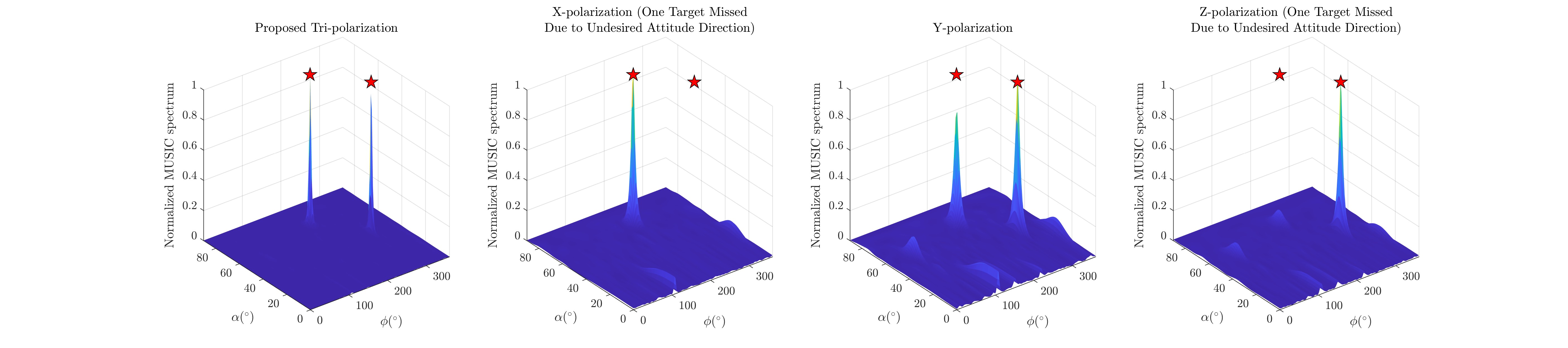}
	\caption{Normalized MUSIC spectrum for CAPA with different polarization}
	\label{fig:spectrumpolarization}
\end{figure*}

Firstly, the MUSIC spectra under different polarization models are illustrated in Fig. \ref{fig:spectrumpolarization}, including the proposed model as well as single-polarization models along the X-, Y-, and Z-axes. The actual DOA values are marked by red stars in this figure. It is observed that for the X- and Z-axis polarization cases, one of the targets cannot be detected, since the target’s dipole orientation is unfavorable and the corresponding signal component is buried in noise. In contrast, the proposed method successfully recovers the peaks of both targets by leveraging the developed tri-polarized spectrum, which demonstrates robustness against specific target orientations. Furthermore, the peaks obtained with the proposed method are obviously sharper than those from the single-polarization models, which is mainly attributed to the joint exploitation of self- and cross-covariance matrices, leading to enhanced DOA resolution.

\subsubsection{Performance Under Different SNR Levels}

\begin{figure}
	\centering
	\includegraphics[width=0.8\linewidth]{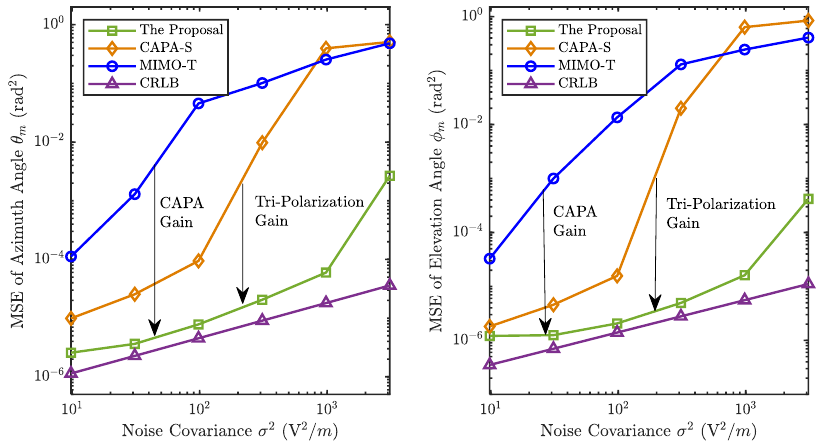}
	\caption{DOA estimation versus noise covariance}
	\label{fig:doanoise}
\end{figure}

The DOA estimation performance under different SNR levels is illustrated in Fig.~\ref{fig:doanoise}, where the MSE of both azimuth and elevation angles are presented versus the noise covariance $\sigma^2$. As shown, the proposed method consistently outperforms the benchmark systems across all SNR regimes. In particular, compared to the single-polarized CAPA (CAPA-S), the proposed tri-polarized design exhibits substantial performance gains, which can be attributed to the exploitation of richer information contained in different polarized components. This observation verifies the effectiveness of polarization diversity in enhancing estimation robustness. Moreover, the superiority of CAPA-based schemes over the conventional MIMO-T approach further demonstrates the advantages of continuous aperture modeling, especially in the low-to-medium SNR range. Finally, the proposed algorithm closely approaches the CRLB, validating its statistical efficiency and near-optimal estimation capability.

\subsubsection{Performance Versus Different Number of Snapshots}

\begin{figure}
	\centering
	\includegraphics[width=0.8\linewidth]{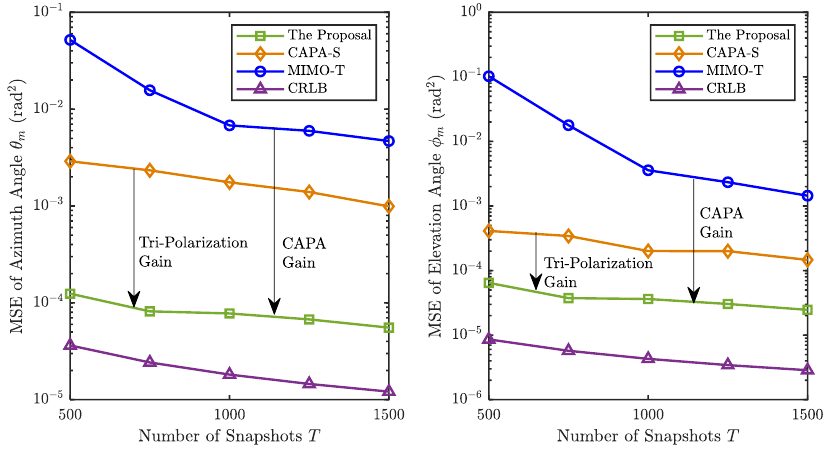}
	\caption{DOA estimation versus number of snapshots}
	\label{fig:doasnapshots}
\end{figure}

The number of snapshots also has a significant impact on DOA estimation accuracy. Fig.~\ref{fig:doasnapshots} illustrates the MSE performance of both azimuth and elevation angle estimations as the number of snapshots $T$ varies within the interval $[500,1500]$. As expected, the estimation error decreases monotonically with more snapshots, since additional temporal samples provide more reliable statistical estimates of the covariance matrices. Among the compared schemes, the proposed tri-polarized CAPA consistently achieves the lowest MSE, approaching the CRLB across all snapshot numbers. In contrast, the single-polarized CAPA (CAPA-S) shows higher estimation errors due to the loss of polarization diversity, while the conventional MIMO-T exhibits the worst performance owing to its limited aperture resolution. These results confirm that the proposed algorithm not only leverages the continuous aperture gain of CAPA but also benefits from the additional polarization gain, leading to superior robustness even in scenarios with relatively few snapshots.

\subsubsection{Performance Versus Different Aperture Sizes}
\begin{figure*}
	\centering
	\includegraphics[width=0.9\linewidth]{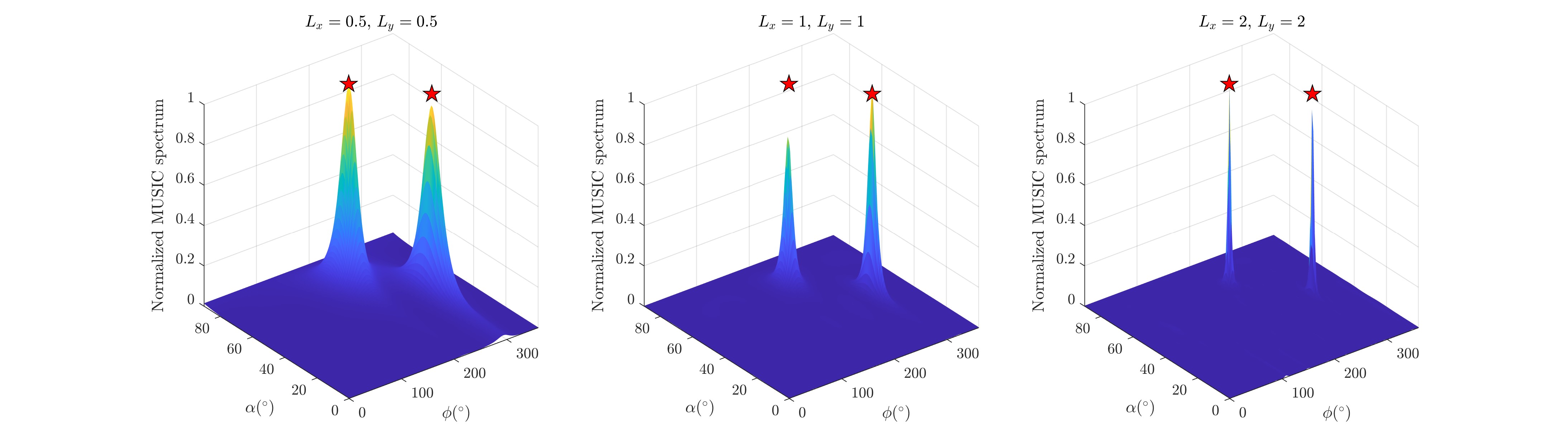}
	\caption{Normalized MUSIC spectrum for CAPA with different aperture size}
	\label{fig:music_pol}
\end{figure*}

The size of CAPA aperture significantly affects the DOA estimation performance, which equals to $L_xL_y$. In Fig. \ref{fig:music_pol}, we illustrate the MUSIC spectrum under three different CAPA sizes: $L_x=L_y=0.5$, $L_x=L_y=1$, and $L_x=L_y=2$. Similarly, we mark the azimuth and elevation angles of targets with red stars in this figure. Obviously, as the increase of the aperture size, the peak of MUSIC spectrum becomes sharper, demonstrating higher DOA resolution and estimation accuracy. This result conforms to the theoretical performance analyses in our previous work \cite{si2025doa}.

\subsubsection{Performance Versus Different Number of Targets}
\begin{figure}
	\centering
	\includegraphics[width=0.9\linewidth]{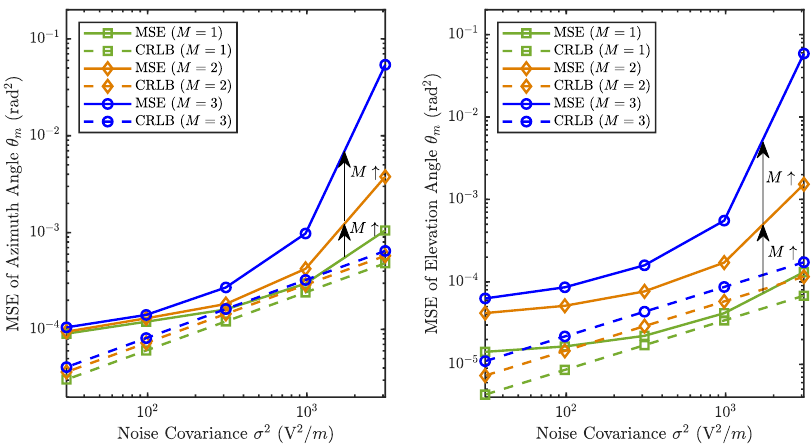}
	\caption{DOA estimation in different number of targets}
	\label{fig:doanumbertargets}
\end{figure}
Fig.~\ref{fig:doanumbertargets} illustrates the DOA estimation performance under different numbers of targets, i.e., $M=1,2,3$. The target positions are set as $\mathbf{p}_1 = [-16, -10, 50]^\mathsf{T}$, $\mathbf{p}_2 = [16, -38, 40]^\mathsf{T}$, and $\mathbf{p}_3=[18, 7.5, 18]^\mathsf{T}$, while the corresponding attitudes are configured as $\bar{\mathbf{q}}_2 = [-0.1, 0.7, 0.7071]^\mathsf{T}$ and $\bar{\mathbf{q}}_3=[-0.8, 0.2, -0.57]^\mathsf{T}$. For the cases of $M=1$ and $M=2$, only the first target and the first two targets are activated, respectively.

As shown in the figure, the MSE of both azimuth and elevation angle estimations increases with the number of targets. This degradation primarily stems from stronger inter-target interference and the reduced effective aperture gain available per target, which make subspace separation more challenging. Despite this, the proposed algorithm achieves consistently satisfactory accuracy, remaining close to the CRLB across all considered target numbers. Even in the case of three targets, the performance gap between the proposed method and the CRLB remains relatively small, which highlights both the robustness and scalability of the proposed approach in multi-target scenarios.

\begin{figure}
	\centering
	\includegraphics[width=0.65\linewidth]{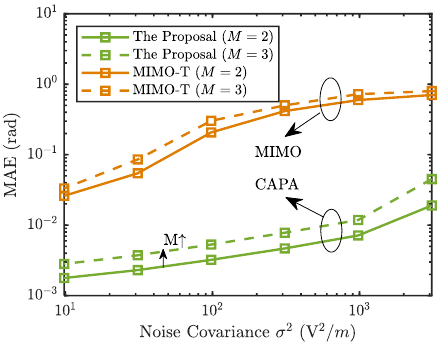}
	\caption{MAE versus noise covariance}
	\label{fig:attitudesnr}
\end{figure}

\subsection{Attitude Estimation Accuracy}
According to the theoretical analyses, when the snapshot information of targets is unknown, only the direction of perdendicular components $\bar{\mathbf{q}}^{\perp \mathbf{z}}_m$ can be estimated. Hence, we first present the MAE of $\bar{\mathbf{q}}^{\perp \mathbf{z}}_m$ to assess the performance of attitude estimation.

\subsubsection{MAE Versus Different SNR}

Fig.~\ref{fig:attitudesnr} presents the MAE performance of attitude estimation versus different noise covariance levels. As discussed earlier, the single-polarized CAPA treats source targets as far-field points, thereby neglecting the polarization-dependent structure of the EM field and making attitude information inestimable. Therefore, only the proposed tri-polarized CAPA algorithm and the MIMO-T scheme are compared here.

It can be observed that the proposed method consistently achieves a significantly lower MAE than the conventional MIMO-T, across all SNR levels. This advantage stems from the continuous aperture gain of CAPA, which provides higher spatial resolution, as well as the tri-polarization gain, which enables effective extraction of attitude-related information. Moreover, as the number of targets increases from $M=2$ to $M=3$, the MAE performance of both schemes degrades slightly due to enhanced inter-target interference. Nevertheless, the proposed algorithm maintains satisfactory estimation accuracy even under low-SNR conditions, demonstrating strong robustness against noise and scalability to multi-target scenarios.

\subsubsection{Visualization of Attitude Ambiguity}
\begin{figure}
	\centering
	\includegraphics[width=0.75\linewidth]{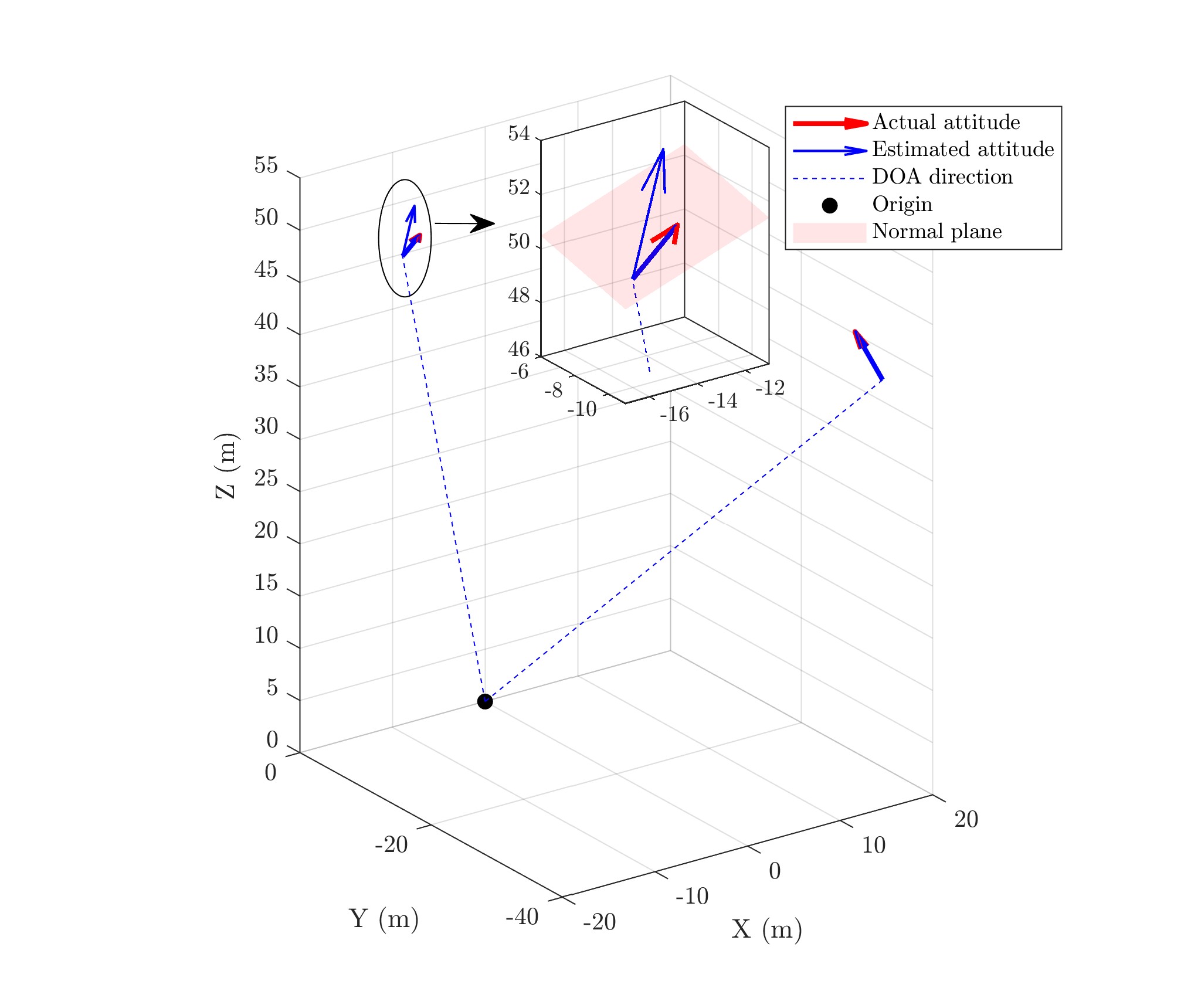}
	\caption{Illustration of attitude ambiguity}
	\label{fig:attitudevisualization}
\end{figure}

Furthermore, the attitude estimation performance with signal snapshots is illustrated in Fig.~\ref{fig:attitudevisualization}. As shown, an attitude ambiguity arises because the directional information along the DOA vector $\bar{\mathbf{z}}_m$ is eliminated by the orthogonal projection $\mathbf{I}-\bar{\mathbf{z}}_m\bar{\mathbf{z}}_m^\mathsf{T}$. Consequently, each target admits two possible attitude solutions that are symmetric with respect to the subspace orthogonal to the DOA direction. This phenomenon stems from the transverse nature of electromagnetic waves in free space and further validate the theoretical analysis.

\section{Conclusions} \label{sec:conclusion}
In this paper, we studied the joint DOA and attitude sensing algorithm for tri-polarized CAPA systems. By establishing an EM information-theoretic model for spatially continuous received signals, we proposed a unified framework that extends CAPA sensing beyond conventional DOA estimation. To address the challenge of subspace decomposition with continuous apertures, an equivalent continuous–discrete transformation method was developed, enabling efficient eigenvalue decomposition. Leveraging both self- and cross-covariance matrices of tri-polarized signals, we further constructed a novel tri-polarized spectrum that significantly enhances DOA estimation resolution and robustness. Based on the obtained DOA estimates, we analyzed the identifiability of target attitude information and designed tailored estimation algorithms for scenarios with and without prior target snapshots. Extensive simulation results demonstrated that the proposed approach achieves near-optimal performance close to the CRLB, while maintaining robustness under various system settings.  In the near future, we will further focus on the sensing of more diverse information of targets in CAPA systems.

\bibliographystyle{IEEEtran}
\bibliography{reference/mybib}

\end{document}